\documentclass{lmcs}
\pdfoutput=1

\usepackage{lastpage}
\lmcsdoi{18}{1}{1}
\lmcsheading{}{\pageref{LastPage}}{}{}%
{Aug.~07,~2020}{Jan.~06,~2022}{}

\usepackage[utf8]{inputenc}
\usepackage[english]{babel}
\usepackage{mathtools}
\usepackage{amsmath}
\usepackage{enumitem}

\usepackage{ebproof}

\usepackage[T1]{fontenc}
\usepackage{amssymb}
\usepackage{stmaryrd}
\usepackage{cmll}

\usetikzlibrary{positioning}

\newcommand*{\N}{\mathbb N}

\newcommand*{\groupG}{\mathcal G}
\newcommand*{\groupH}{\mathcal H}
\newcommand*{\groupK}{\mathcal K}
\newcommand*{\eqdef}{\coloneqq}
\newcommand*{\recdef}{\Coloneqq}
\newcommand{\comments}[1]{}

\newcommand*\sums[1]{\N[#1]}
\newcommand*\cons{\dblcolon}
\newcommand*\length[1]{\lvert #1\rvert}
\newcommand*\vectors[2][\Q^+]{#1\langle #2\rangle}
\newcommand*\rappl[3][]{#1\langle #2#1\rangle #3}
\newcommand*\subst[4][]{#2 #1[#4/#3 #1]}
\newcommand*\nsubst[3]{\partial_{#2} #1\cdot #3}
\newcommand*\linj[1]{#1\oplus\mathord\bullet}
\newcommand*\rinj[1]{\mathord\bullet\oplus#1}
\newcommand*\lact[3][]{#1[ #2#1] #3}

\newcommand*\ndTerms{\Lambda_\oplus}
\newcommand*\ndApprox{\Lambda_\oplus^\bot}
\newcommand*\resTerms{\Delta_\oplus}
\newcommand*\resMonomials{\resTerms^\oc}
\newcommand*\resExprs{\resTerms^{(\oc)}}
\newcommand*\resTermSums{\sums{\resTerms}}
\newcommand*\resMonomialSums{\sums{\resMonomials}}
\newcommand*\resExprSums{\sums{\resExprs}}
\newcommand*\resTermVectors{\vectors{\resTerms}}
\newcommand*\resMonomialVectors{\vectors{\resMonomials}}

\newcommand*\rigidTerms{D}
\newcommand*\rigidMonomials{\rigidTerms^\oc}
\newcommand*\rigidExprs{\rigidTerms^{(\oc)}}
\newcommand*\softify[2][]{#1\lVert{#2}#1\rVert}

\newcommand*{\tayexp}[1]{{#1}^*}
\newcommand*{\taysup}[1]{T(#1)}
\newcommand*{\eBTs}{\mathcal N}
\newcommand*{\eBT}[1]{\eBTs(#1)}
\newcommand*{\BT}[1]{\mathit{BT}(#1)}
\newcommand*{\NF}[1]{\mathit{NF}(#1)}
\newcommand*{\D}{\mathbb D}
\newcommand*{\vD}{\vec{\D}}
\newcommand*{\Dm}{\D^\oc}
\newcommand*{\G}{\D^{(\oc)}}
\newcommand*{\card}{\mathit{Card}}
\newcommand*{\wcomp}[2]{\mathcal Q_{#1}(#2)}
\newcommand*{\wcompk}[3]{\mathcal Q_{#1}^{#2}(#3)}
\newcommand*{\permutations}[1]{\mathfrak S_{#1}}
\newcommand*{\catP}{\mathbb P}
\newcommand*{\St}{\mathit{St}}
\newcommand*{\qSt}{\mathit{St}_{\cong}}
\newcommand*{\subSt}[1]{\mathcal{H}_{#1}}
\newcommand*{\preqSt}[1]{\mathcal{K}_{#1}}
\newcommand*{\supp}{\mathit{supp}}

\newcommand*\definitive[1]{\emph{#1}}

\newcommand*\restr[2]{#1\mathord\upharpoonright_#2}

\title[
	Taylor exp.\ of \texorpdfstring{$\lambda$}{lambda}-terms and the groupoid of rigid approximants
]{
	On the Taylor expansion of \texorpdfstring{$\lambda$}{lambda}-terms and the groupoid structure of their rigid approximants
}
\titlecomment{
	This article is an extended version of the abstract~\cite{ov:groupoid-tlla}
	presented at Linearity \& TLLA'2018.
}

\author{Federico Olimpieri}
\author{Lionel Vaux Auclair}
\address{Aix-Marseille Univ, CNRS, I2M, Marseille, France}
\email{\{federico.olimpieri,lionel.vaux\}@univ-amu.fr}

\ACMCCS{\textit{Theory of computation~Denotational semantics; }
Theory of computation~Lambda calculus;
Theory of computation~Linear logic;}

\begin{document}

\begin{abstract}
	We show that the normal form of the Taylor expansion of a $\lambda$-term is
	isomorphic to its Böhm tree, improving Ehrhard and Regnier's original proof
	along three independent directions.

	First, we simplify the final step of the proof by following the left
	reduction strategy directly in the resource calculus, avoiding to introduce
	an abstract machine \emph{ad hoc}.

	We also introduce a groupoid of permutations of copies of arguments in a
	rigid variant of the resource calculus, and relate the coefficients of
	Taylor expansion with this structure, while Ehrhard and Regnier worked
	with groups of permutations of occurrences of variables.

	Finally, we extend all the results to a nondeterministic setting: by
	contrast with previous attempts, we show that the uniformity property that
	was crucial in Ehrhard and Regnier's approach can be preserved in this
	setting.
\end{abstract}

\maketitle

\section{Introduction}

\subsection{Quantitative semantics}
The field of quantitative semantics, in the sense originally introduced by
Girard~\cite{girard:qs}, is currently very lively within the linear
logic community and beyond.
The basic idea is to interpret $\lambda$-terms as generalized power
series, associated with analytic maps --- instead of continuous maps,
\emph{à la} Scott.
The concept predates linear logic, and in fact it provided the foundations for it,
\emph{via} its simpler, qualitative counterpart: coherence spaces~\cite{girard:ll}.
It was later revisited, e.g.\ by Lamarche~\cite{lamarche:quantitative}
and Hasegawa~\cite{hasegawa:applications}, to provide a denotational
interpretation of linear logic proofs as matrices;
but the current momentum originates in the more recent introduction by Ehrhard~\cite{ehrhard:fs} of models of linear logic, based on a particular class of
topological vector spaces, and thus accommodating differentiation.

In that setting, the analytic maps associated with $\lambda$-terms are also
smooth maps, \emph{i.e.} they are infinitely differentiable.
This led to the differential extensions of $\lambda$-calculus~\cite{er:difflamb} and linear logic~\cite{er:diffnets} by Ehrhard and Regnier.
The keystone of this line of work is an analogue of the
Taylor expansion formula, which allows to translate terms (or proofs)
into infinite linear combinations of finite approximants~\cite{er:tay}:
in the case of $\lambda$-calculus, those approximants are the terms of a
resource calculus, in which the copies of arguments of a function must be
provided explicitly, and then consumed linearly, instead of duplicated or
discarded during reduction.

This renewed approach to quantitative semantics served as the basis of a
considerable amount of recent work: either as a framework for denotational
models accommodating linear combinations of maps~\cite[\emph{etc.}]{lmmp:weighted,laird:fp,tao:gen,ong:qs},
possibly in contexts where sums are constrained to a particular form,
such as the probabilistic setting~\cite[\emph{etc.}]{de:pcs,tao:pdist};
or as a tool for characterizing computational properties of programs
\emph{via} those of their approximants~\cite[\emph{etc.}]{mp:bohm,ptv:sn,dll:teplt,bm:taylor}.

Indeed, by contrast with denotational semantics, resource approximants
retain a dynamics, albeit very simple and finitary: the size of terms is
strictly decreasing under reduction.
The seminal result relating the reduction of $\lambda$-terms with that of their
approximants is the commutation between Taylor expansion and normalization:
Ehrhard and Regnier have shown that the Taylor expansion $\tayexp M$ of a
$\lambda$-term $M$ can always be normalized, and that its normal form is
nothing but the Taylor expansion of the Böhm tree $\BT{M}$ of $M$~\cite{er:tay,er:boh}.
In particular, the normal form of Taylor expansion defines a proper
denotational semantics.

\subsection{Contributions}
Ehrhard and Regnier's proof of the identity $\tayexp{\BT M}=\NF{\tayexp M}$
can be summed up as follows:
\begin{enumerate}[label={Step \arabic*:},align=left,ref={\arabic*}]
	\item\label{step:m}
		The non-zero coefficients of resource terms in $\tayexp M$
		do not depend on $M$. More precisely,
		we can write $\tayexp M=\sum_{s\in\taysup{M}} \frac {1}{m(s)} s$,
		where $\taysup{M}$ is the support set of Taylor expansion
		and $m(s)$ is an integer coefficient depending only on
		the resource term $s$.
	\item\label{step:clique}
		The set $\taysup{M}$ is a clique for the coherence relation obtained by setting
		$s\coh s'$ iff $s$ and $s'$ differ only by the multiplicity of arguments in applications.
	\item\label{step:disjoint}
		The respective supports of $NF(s)$ and $NF(s')$ are disjoint
		whenever $s\coh s'$ and $s\not=s'$.
		Then one can set
		$NF(\tayexp M)=\sum_{s\in\taysup{M}} \frac {1}{m(s)} NF(s)$,
		the summands being pairwise disjoint.
	\item\label{step:coefNF}
		If $s$ is uniform, \emph{i.e.} $s\coh s$,
		and $t$ is in the support of $NF(s)$
		(the normal form of $s$, which is a finite sum of resource terms)
		then $m(t)$ divides $m(s)$ and the coefficient of $t$ in $NF(s)$ is $\frac{m(s)}{m(t)}$.
	\item\label{step:taysup}
		By Step~\ref{step:m},
		$\tayexp{\BT{M}}=\sum_{t\in \taysup{\BT{M}}}\frac{1}{m(t)} t$.
		To deduce the identity $\tayexp{\BT M}=\NF{\tayexp M}$ from the previous results,
		it is then sufficient to prove that
		$t\in \taysup{\BT{M}}$ iff there exists $s\in\taysup M$
		such that $t$ is in the support of $NF(s)$.
\end{enumerate}
The first two steps are easy consequences of the definitions.
For Step~\ref{step:disjoint}, it is sufficient to follow a well chosen
normalization strategy, and check that it preserves coherence and that if two
coherent terms share a reduct then they are equal~\cite[Section~3]{er:tay}.
Step~\ref{step:coefNF} relies on a careful investigation
of the combinatorics of substitution in the resource calculus:
this involves an elaborate argument about the structure of particular
subgroups of the group of permutations of variable occurrences~\cite[Section~4]{er:tay}.
Finally, Ehrhard and Regnier establish Step~\ref{step:taysup}
by relating Taylor expansion with execution in an abstract machine~\cite{er:boh}.

In the present work, we propose to revisit this seminal result, along three directions.
\begin{enumerate}[label={(\roman*)}]
	\item\label{contrib:taysup}
		We largely simplify Step~\ref{step:taysup}, relying on a
		technique introduced by the second author~\cite{vaux:taylornf}.
		We consider the hereditary head reduction strategy (a slight variant
		of leftmost reduction, underlying the construction of Böhm trees)
		and show that it can be simulated directly in the resource calculus,
		through Taylor expansion.
		We thus avoid the intricacies of an abstract machine with resource state.
	\item\label{contrib:nd}
		We extend all the results to a model of nondeterminism,
		introduced as a formal binary choice operator in the calculus.
		By contrast with previous proposals to nondeterminism from Ehrhard~\cite{ehrhard:finres}, or Pagani, Tasson and Vaux Auclair~\cite{ptv:taylorsn,vaux:taylornf}, we show that uniformity
		can still be relied upon, provided one keeps track of choices
		in the resource calculus:
		the coherence associated with nondeterministic choice is then that of the
		\emph{with} connective ($\with$) of linear logic.
	\item\label{contrib:groupoid}
		We analyse coefficients in the Taylor expansion by introducing a groupoid
		whose objects are rigid resource terms, \emph{i.e.} resource terms in which
		multisets of arguments are replaced with lists,
		and whose isomorphisms are permutation terms, \emph{i.e.} terms equipped
		with permutations that act on lists of arguments.
		This is more in accordance with the intuition that $m(s)$ is
		the number of permutations of arguments that leave $s$
		(or rather, any rigid representation of $s$) invariant:
		Ehrhard and Regnier rather worked on permutations of variable
		occurrences, which allowed them to consider groups rather than a groupoid.
\end{enumerate}
Although we implement all three contributions together,
they are essentially independent of each other.
Indeed, the simplification of Step~\ref{step:taysup} brought by our
contribution~\ref{contrib:taysup} only concerns the compatibility of Taylor
expansion with normalization at the level of support sets, which does not
involve coefficients; and it does not rely on uniformity,
so its extension to nondeterministic superpositions is straightforward.

Moreover, while our contribution~\ref{contrib:nd} enables us to enforce the
uniformity condition of Steps~\ref{step:clique} to~\ref{step:coefNF}
in presence of a choice operator, it also ensures that distinct branches of a
choice have disjoint supports in the Taylor expansion.
This treatment of nondeterminism makes it completely transparent in the
computation of coefficients.
In particular, one could straightforwardly extend all steps of Ehrhard and
Regnier's proof in this setting, \emph{ceterit paribus}.

Our contribution~\ref{contrib:groupoid} is thus not needed for that endeavour:
it only offers an alternative viewpoint on the combinatorics of substitution
and normalization in the resource calculus, in a uniform setting.
Nonetheless, we consider it to be the main contribution of the paper, precisely
because of the new light it sheds on this dynamics, which in turn reveals
possible connections with other approaches.

\subsection{Scope and related works}
Our contribution~\ref{contrib:taysup} establishes that, although it is
interesting in itself, Ehrhard and Regnier's study of the relationship
between elements in the Taylor expansion of a term and its execution in an
abstract machine is essentially superfluous for proving the commutation theorem.

Barbarossa and Manzonetto have independently
proposed another technique which amounts to show that any reduction from an
element of $\taysup M$ can be completed into a sequence of reductions
simulating a $\beta$-reduction step~\cite[Section 4.1]{bm:taylor}.
The strength of our own proposal is that, rather than a mere simulation result,
we establish a commutation on the nose:
hereditary head reduction commutes with Taylor expansion,
at the level of supports.
Moreover, the Böhm tree of a $\lambda$-term is the limit of its hereditary head
reducts, which ensures that this commutation extends to normalization (Step~\ref{step:taysup}).
The same path was followed by Dal Lago and Leventis~\cite{dll:teplt} for the
probabilistic case.
Let us mention that the commutation with hereditary head reduction
actually holds not only at the level of supports, but also taking coefficients
into account~\cite{vaux:taylornf}, in the more general setting of the algebraic
$\lambda$-calculus~\cite{vaux:alg} and without any additional condition:
then, whenever the convergence of the sum defining the normal form of Taylor
expansion is assured, the main commutation theorem ensues directly.
This offers an alternative to the method of Ehrhard and Regnier that is the
focus of the present paper.

As stated before, our proposal~\ref{contrib:nd} to restore uniformity in a
nondeterministic setting is valid only because the resource calculus keeps a
syntactic track of choices.
The corresponding constructors are exactly those used by Tsukada, Asada and
Ong~\cite{tao:gen} who were interested in identifying equivalent execution paths
of nondeterministic programs, but those authors do not mention, nor rely upon
any coherence property: this forbids Steps~\ref{step:m} to~\ref{step:coefNF} and, instead, they depend on infinite sums of arbitrary
coefficients to be well defined.
By contrast, Dal Lago and Leventis have independently proposed nearly the same
solution as ours~\cite[Section 2.2]{dll:teplt}, with only a minor technical
difference in the case of sums.

The previous two proposals~\ref{contrib:taysup} and~\ref{contrib:nd} may be
considered as purely technical improvements of the state of the art in the
study of Taylor expansion.
What we deem to be the most meaningful contribution of the present paper is our
study of the groupoid of rigid resource terms.
This provides us with a new understanding of the coefficients
in the Taylor expansion of a term, in which we can recast the proof of the
commutation theorem, especially Step~\ref{step:coefNF}:
apart from this change of focus, the general architecture of our approach does
not depart much from that of Ehrhard and Regnier,
but we believe the obtained combinatorial results are closer to the original
intuition behind the definition of $m$.
In fact, a notable intermediate result (Lemma~\ref{lem:functor}, p.\pageref{lem:functor})
is that the function that maps each permutation term
to the permutation it induces on the occurrences of a fixed variable is functorial:
one might understand Ehrhard and Regnier's proof of Step~\ref{step:coefNF}
as the image of ours through that functor.
Moreover, our study suggests interesting connections with otherwise
independent approaches to denotational semantics based on generalized species
of structures~\cite{fiore:esp,tao:gen} and rigid intersection type systems~\cite{mazza:pol}.

It is indeed most natural to compare our proposals to the line of work of
Tsukada, Asada and Ong~\cite{tao:gen,tao:pdist}.
On the one hand, Tsukada \emph{et al.}
thrive to develop an abstract understanding of
reduction paths in a nondeterministic $\lambda$-calculus.
They are led to consider a polyadic calculus \emph{à la} Mazza~\cite{mazza:affine,mazza:pol} with syntactic markers for nondeterministic choice,
moreover obeying linearity, typing and $\eta$-expansion
constraints:
in particular, in that polyadic setting, $\lambda$-abstractions bind lists of variables,
each bound variable occurring exactly once.
Then, to each simple type, they associate a groupoid of intersection types:
an isomorphism in this groupoid acts on polyadic rigid terms by permuting
variables bound in abstractions and lists of arguments in applications,
in such a way that terms in its source intersection type yield terms in its target.
They show that the obtained collection of groupoids form a bicategorical model
of the simply typed $\lambda\mathbf Y$-calculus, the interpretation being given
by a polyadic rigid variant of Taylor expansion.
This interpretation is moreover isomorphic to the one
obtained in generalized species of structures~\cite{fiore:esp}.

On the other hand, our contribution~\ref{contrib:nd} shows that Ehrhard and Regnier's technique
can already be adapted to the same kind of nondeterminism as the one
considered by Tsukada \emph{et al.}, without introducing any
new concept.
Also, besides having markers for nondeterministic choice, the only difference
between our rigid terms and the ordinary resource terms is that arguments are
linearly ordered: we do not consider a polyadic version.
In fact, the same rigid terms were already used by Tsukada
\emph{et al.} as \emph{intermediate representations} of resource terms, in
order to recover Ehrhard and Regnier's commutation theorem as a by-product of
their construction~\cite[Section VI]{tao:gen}.
Moreover, our permutation terms are similar to their typed isomorphisms
and this suggests directions for further investigations.

A natural follow-up to the present work would thus be to explore possible
variations on our groupoid of permutation terms, and in particular adapt it to
a polyadic setting, also taking free variables into account.
We expect this study to yield a bicategorical model of the pure, untyped
$\lambda$-calculus, similarly induced by rigid Taylor expansion \emph{à la}
Tsukada--Asada--Ong. 
Then potential connexions between the obtained model and the construction of
various reflexive objects in the bicategory of generalized species of
structures~\cite[Section~6.2]{fiore:esp} should be investigated.

Another possible route to the untyped setting, actively developed by
the first author, is to construct a category satisfying a domain-like equation
in the model of generalized species~\cite{olimpieri:intersection}.
The objects in this category are very much like intersection types, except that
the usual identities between types (commutativity and, possibly, idempotency)
are made explicit as morphisms,
which allows to develop a bicategorical treatment of intersection type systems.

\subsection{Structure of the paper}

In the very brief Section~\ref{sect:group},
we review some results from group theory that will
be useful later.

In Section~\ref{sect:nd} we extend the ordinary untyped
$\lambda$-calculus with a generic nondeterministic choice operator,
and present its operational semantics, inspired from that
of the algebraic $\lambda$-calculus, as well as the corresponding notion of
(non extensional) Böhm trees.

Section~\ref{sect:resource} recalls and adapts the definitions of the resource
calculus and Taylor expansion. We obtain Step~\ref{step:clique} as a
straightforward consequence of the definitions and Step~\ref{step:taysup}
by showing that the support of Taylor expansion is compatible with hereditary
head reduction --- this is our contribution~\ref{contrib:taysup}.
We moreover complete Step~\ref{step:m},
making prominent the rôle played by permutations acting on lists of resource
terms.

Section~\ref{sect:rigid} is the core of the paper,
developing our main contribution~\ref{contrib:groupoid}:
we introduce the rigid version of resource terms,
and the isomorphisms between them, given by permutation terms;
then we explore the relationship between the groupoid thus formed
and the combinatorics of Taylor expansion.
We first show that the coefficient $m(s)$
is nothing but the cardinality of the group of automorphisms
of any rigid version of $s$.
Then we study the structure of permutation terms between substitutions,
first in the general case, then in the uniform case
--- which is allowed in our nondeterministic setting
thanks to our contribution~\ref{contrib:nd}.
We leverage the obtained results to determine the coefficient
of any resource term in the symmetric multilinear substitution
associated with a reduction step issued from a uniform redex.

The final Section~\ref{sect:commutation} builds on the study of rigid resource
terms and permutation terms to achieve
Steps~\ref{step:disjoint} and~\ref{step:coefNF}.
We conclude the paper with the commutation theorem.

\section{Some basic facts on groups and group actions}%
\label{sect:group}

Let $\groupG$ be a group, $X$ be a set,
and write $(g,a)\in \groupG\times X\mapsto \lact{g}{a}\in X$
for a left action of $\groupG$ on $X$.
If $a\in X$, then the \definitive{stabilizer}
of $a$ under this action is
$\St(a)\eqdef \{g\in\groupG \mid \lact{g}{a}=a\}$,
which is a subgroup of $\groupG$
(also called the isotropy group of $a$);
and the \definitive{orbit} of $a$ is the set
$\lact{\groupG}a\eqdef \{\lact {g}{a}\mid g\in\groupG\}\subseteq X$.
If $\groupH,\groupK\subseteq \groupG$,
we write $\groupH\groupK\eqdef \{hk\mid h\in\groupH,\ k\in\groupK\}$.
If $f:X\to Y$, $X'\subseteq X$ and $Y'\subseteq Y$
we write $f(X')\eqdef \{ f(x)\mid x\in X'\}$
and $f^{-1}(Y')\eqdef \{ x\mid f(x)\in Y'\}$.

Assuming that $\groupG$ is finite,
the following three facts are standard results of group theory.
\begin{fact}%
	\label{fact:bijG}
	For any $a\in X$,
	\[
		\card(\lact{\groupG}a)=\dfrac{\card(\groupG)}{\card(\St(a))}
		\quad .
	\]
\end{fact}
\begin{proof}%
	\cite[Proposition 5.1]{lang:algebra}.
\end{proof}

\begin{fact}%
	\label{fact:groupProduct}
	Let $\groupH$ and $\groupK$ be any subgroups of $\groupG$. Then
	\[
		\card(\groupH\groupK)
		= \dfrac{\card( \groupH)\card(\groupK)}
		{\card ( \groupH \cap \groupK)}
		\quad .
	\]
\end{fact}
\begin{proof}%
	\cite[§(3.11)]{suzuki:grouptheoryI}.
\end{proof}
\begin{fact}\label{fact:cardQuotient}
	Let $f:\groupG\to\groupH$ be a group homomorphism
	and $\groupK$ be a subgroup of $\groupH$. Then
	\[
		\frac{\card(\groupG)}{\card(f^{-1}(\groupK))}
		=
		\frac{\card(f(\groupG))}{\card(f(\groupG)\cap \groupK)}
		\quad .
	\]
\end{fact}
\begin{proof}
	By the theorem of correspondence under homomorphisms~\cite[Theorem 5.5 (1)]{suzuki:grouptheoryI},
  observing that $f(\groupG)\cap \groupK=f(f^{-1}(\groupK))$.
\end{proof}

\section{A generic nondeterministic \texorpdfstring{$\lambda$}{lambda} calculus}%
\label{sect:nd}

\subsection{\texorpdfstring{$\lambda_\oplus$}{lambdaplus}-terms}

We consider a nondeterministic version of the $\lambda$-calculus in a pure, untyped setting.
The terms are those of the pure $\lambda$-calculus, augmented with a binary
operator $\oplus$ denoting a form of nondeterministic superposition:\footnote{
Throughout the paper, we use a self explanatory if not standard variant of BNF notation
for introducing syntactic objects: here we define the set $\ndTerms$
as that inductively generated by variables, $\lambda$-abstraction, application and sum,
and we will denote terms using letters among $M,N,P,Q$, possibly with
sub- and superscripts.
}
\[ \ndTerms \ni M,N,P,Q \recdef x \mid \lambda x. M \mid MN \mid M \oplus N. \]
As usual $\lambda_{\oplus}$-terms are considered up to renaming bound variables,
and we write $M[N/x]$ for the capture avoiding substitution of $N$ for $x$ in $M$.
We give precedence to application over abstraction, and to abstraction over $\oplus$,
and moreover associate applications on the left,
so that we may write $\lambda x.MNP\oplus Q$ for $(\lambda x.((MN)P))\oplus Q$.
We write $\lambda \vec x.M$ for a term of the form $\lambda x_1.\cdots\lambda x_n.M$
(possibly with $n=0$).

Rather than specifying the computational effect of $\oplus$ explicitly,
by reducing $M\oplus M'$ to either $M$ or $M'$,
we consider two reductions rules
\[
	(M\oplus N)P\to MP\oplus NP \qquad\text{and}\qquad \lambda x. ( M \oplus N) \to \lambda x. M \oplus \lambda x. N
\]
in addition to the $\beta$-reduction rule.
  This is in accordance with most of the literature associated with
	the Taylor expansion of $\lambda$-terms~\cite{er:difflamb,ehrhard:finres,ptv:taylorsn,vaux:taylornf}
	and quantitative denotational semantics~\cite{ehrhard:fs},
	where nondeterministic choice is modelled by the sum
	of denotations:
  rather than the current state of a nondeterministic computation, a
  term represents a superposition of possible results.\footnote{
	In fact, only the rule $(M\oplus N)P\to MP\oplus NP$ is really
	necessary in order to enable the potential redexes that
	can occur if $M$ or $N$ is an abstraction:
	in the setting of quantitative semantics,
  term application is left-linear.
	The other reduction rule can be derived in case one admits extensionality in
	the models or the $\eta$-rule in the calculus (here we don't, though):
  having it in the calculus means that we follow a \emph{call-by-name}
  interpretation of nondeterministic evaluation,
  which amounts to $\lambda$-abstraction being linear~\cite{adcptv:cbvcbnvectorial}.
	The results of the paper could be developed similarly without it.
	We chose to keep it nonetheless, because it simplifies the
	underlying theory of Böhm trees and allows us to obtain
	Ehrhard and Regnier's results~\cite{er:tay,er:boh} as a particular case of
	our own.
}
  In particular, this approach allows to keep standard rewriting notions and
  techniques such as confluence, standardization, etc.
Formally, $\to$ is defined inductively by the inference rules
of Figure~\ref{fig:red}: we simply extend the three base cases
contextually.

Observe that neither the definition of terms nor that of reduction
make the choice operator commutative, associative nor idempotent:
e.g., $x\oplus y$ and $y\oplus x$ are two distinct normal forms.
It is possible to extend the reduction relation to validate the
structural properties associated with various kinds of superpositions
(plain nondeterministic choice, probabilistic choice or a more
general quantitative superposition) while retaining good
rewriting properties: we refer the reader to the work of
Leventis~\cite{leventis:standardization} for an
extensive study of this approach.

By contrast, for our purposes, it is essential to keep $\oplus$ as a free
binary operator: following Tsukada, Asada and Ong~\cite{tao:gen},
we keep track of the branching structure of choices along the reduction.
This information will be reflected in the Taylor expansion
to be introduced in Section~\ref{sect:resource}:
this is the key to recover the uniformity property while
allowing for nondeterministic superpositions of terms.

\begin{figure}[t]
	\begin{gather*}
		\begin{prooftree}
			\infer0{(\lambda x.M)N\to M[N/x]}
		\end{prooftree}
		\qquad
		\begin{prooftree}
			\infer0{(M\oplus N)P\to MP\oplus NP}
		\end{prooftree}
		\qquad
		\begin{prooftree}
			\infer0{\lambda x. ( M \oplus N) \to \lambda x. M \oplus \lambda x. N}
		\end{prooftree}
		\\[1ex]
		\begin{prooftree}
			\hypo{M\to M'}
			\infer1{\lambda x.M\to\lambda x.M'}
		\end{prooftree}
		\quad
		\begin{prooftree}
			\hypo{M\to M'}
			\infer1{MN\to M'N}
		\end{prooftree}
		\quad
		\begin{prooftree}
			\hypo{M\to M'}
			\infer1{NM\to NM'}
		\end{prooftree}
		\quad
		\begin{prooftree}
			\hypo{M\to M'}
			\infer1{M\oplus N\to M'\oplus N}
		\end{prooftree}
		\quad
		\begin{prooftree}
			\hypo{M\to M'}
			\infer1{N\oplus M\to N\oplus M'}
		\end{prooftree}
	\end{gather*}
	\caption{Reduction rules of the $\lambda_\oplus$-calculus}%
	\label{fig:red}
\end{figure}

In fact we will not really consider the reduction relation $\to$ in the present paper,
and rather focus on the \definitive{hereditary head reduction strategy} obtained
by defining the function $L:\Lambda_\oplus\to\Lambda_\oplus$
inductively as follows:
\begin{align*}
	L(M\oplus N)
	&\eqdef L(M)\oplus L(N)
	\\
	L(\lambda\vec x.\lambda y.(M\oplus N))
	&\eqdef \lambda\vec x.(\lambda y.M\oplus \lambda y.N)
	\\
	L(\lambda\vec x.(M \oplus N)P Q_1\cdots Q_k)
	&\eqdef \lambda\vec x.(MP \oplus NP)Q_1\cdots Q_k
	\\
	L(\lambda\vec x.y Q_1\cdots Q_k)
	&\eqdef \lambda\vec x.y L(Q_1)\cdots L(Q_k)
	\\
	L(\lambda\vec x.(\lambda y. M) N Q_1\cdots Q_k)
	&\eqdef \lambda\vec x.\subst{M}{y}{N}Q_1\cdots Q_k
	\quad.
\end{align*}
Observe that this definition is exhaustive because any term in $\Lambda_\oplus$
is either of the form $M\oplus N$ or of the form
$\lambda\vec x.\lambda y.(M\oplus N)$
or of the form $\lambda\vec x.RQ_1\cdots Q_k$ with
$R=(\lambda y. M) N$
or
$R=(M \oplus N)P$
or
$R=y$.

It should be clear that $M\to^* L(M)$ and that $L(M)=M$ whenever
$M$ is normal\footnote{
	If one considers $\oplus$ as a nondeterministic choice operator,
	normalizability is meant in its \emph{must} flavour here.
	Indeed, we do not perform the choice within the reduction relation itself,
	so $M\oplus N$ is normal iff $M$ and $N$ both are.
}
but the converse does not necessarily hold.
It can moreover be shown that any normalizable term $M$ reaches its normal form
by repeatedly applying the function $L$, for instance by adapting the
standardization techniques of Leventis~\cite{leventis:phd,leventis:standardization},
but this is not the focus of the present paper.
Indeed, we are only interested in the construction of Böhm trees, and we rely
on the fact that the Böhm tree of a term $M$ can be understood as the limit of
the sequence $(L^n(M))_{n\in\N}$, in a sense that we detail below.
In particular, $M$ and $L(M)$ have the same Böhm tree
(Lemma~\ref{lem:BTL}).

\subsection{B\"{o}hm trees}

We first define the set $\ndApprox$ of \definitive{term approximants} as follows:
\[ \ndApprox \ni M,N,P,Q \recdef \bot \mid x \mid \lambda x. M \mid MN \mid M\oplus N \]
then we consider the least partial order
$ {\leq} \subseteq \ndApprox\times \ndApprox$
that is compatible with syntactic constructs
and such that $\bot\le M$ for each $M\in\ndApprox$.
Formally, $\leq$ is defined inductively
by the rules of Figure~\ref{fig:leq}.

\begin{figure}[t]
	\begin{gather*}
		\begin{prooftree}
			\infer0[]{ \bot \leq M }
		\end{prooftree}
		\qquad
		\begin{prooftree}
			\infer0[]{ M \leq M}
		\end{prooftree}
		\qquad
		\begin{prooftree}
			\hypo{ M \leq N}
			\hypo{ N \leq P}
			\infer2[]{ M \leq P}
		\end{prooftree}
		\\[1ex]
		\begin{prooftree}
			\hypo{ M \leq M'}
			\infer1[]{ \lambda x. M \leq \lambda x. M' }
		\end{prooftree}
		\qquad
		\begin{prooftree}
			\hypo{ M \leq M'}
			\hypo{ N \leq N'}
			\infer2[]{ MN \leq M'N' }
		\end{prooftree}
		\qquad
		\begin{prooftree}
			\hypo{ M \leq M'}
			\hypo{ N \leq N'}
			\infer2[]{ M \oplus N \leq M' \oplus N' }
		\end{prooftree}
	\end{gather*}
\caption{The approximation order on $\ndApprox$.}%
\label{fig:leq}
\end{figure}

The set $\eBTs\subset \ndApprox$ of
\definitive{elementary Böhm trees} is
the least set of approximants such that:
\begin{itemize}
	\item $\bot\in\eBTs$;
	\item $\lambda \vec x. x N_{1} \cdots N_{n} \in\eBTs$
    as soon as $N_{1}, \dotsc , N_{n}\in\eBTs$;\footnote{
      Here the sequence $\lambda\vec x$ of abstractions can be empty,
      and we can have $n=0$,
      in which case the body of the term is just the head variable.
    } and
	\item $N_{1}\oplus N_{2} \in\eBTs$
		as soon as $N_{1}, N_{2}\in\eBTs$.
\end{itemize}
For each $\lambda_\oplus$-term $M$, we construct an elementary B\"{o}hm tree $\eBT M$
as follows:
\begin{align*}
	\eBT{M \oplus N}
	&\eqdef
	\eBT M \oplus \eBT N
	\\
	\eBT{\lambda \vec x. x Q_1 \cdots Q_k}
	&\eqdef
	\lambda \vec x. x \eBT{Q_1} \cdots \eBT{Q_k}
	\\
	\eBT{M}
	&\eqdef
	\bot
	\qquad\text{in all other cases.}
\end{align*}

\begin{lem}\label{lem:increasingBT}
For any $M\in\ndTerms$,
$\eBT M\le \eBT{L(M)}$.
\end{lem}
\begin{proof}
By induction on $M$.
If $M=M_1\oplus M_2$ then
$\eBT M=\eBT{M_1}\oplus\eBT{M_2}$
and
$L(M)=L(M_1)\oplus L(M_2)$,
hence
$\eBT{L(M)}=\eBT{L(M_1)}\oplus\eBT{L(M_2)}$
and we conclude by induction hypothesis.
The case $M=\lambda \vec x.x Q_1\cdots Q_k$ is similar.
Otherwise,
$\eBT M=\bot\le \eBT{L(M)}$.
\end{proof}

Hence for a fixed $\lambda_\oplus$-term $M$, the sequence
$(\eBT{L^n(M)})_{n\in\N}$ is increasing,
and we call its downwards closure in $\eBTs$ \definitive{the Böhm tree of $M$},
which we denote by $\BT M$:
\emph{i.e.} we set $\BT M \eqdef\{ N\in\eBTs\mid \exists n\in\N,\ N\le \eBT{L^n(M)} \}$.
\begin{exa}%
  \label{exa:BT}
  Let $M=\Theta\lambda y.(y\oplus x)$
  where $\Theta$ is Turing's  fixpoint combinator,
  so that $L^3(M)=M\oplus x$.
  We can think of $\BT{M}$ as the infinite tree
  $((\cdots\oplus x)\oplus x)\oplus x$:
  formally, $\BT{M}=\{\bot\oplus nx\mid n\in\N \}$ where
  we define inductively $M\oplus 0N=M$ and $M\oplus(n+1)N=(M\oplus nN)\oplus N$.
\end{exa}

It could be shown that Böhm trees define a denotational semantics:
if $M\to M'$ then $BT(M)=BT(M')$.\footnote{
	Again, this would require the adaptation of standardization techniques
	to $\lambda_\oplus$, similar to those developed by
	Leventis for the probabilistic $\lambda$-calculus~\cite{leventis:standardization}.
}
Here we just observe that Böhm trees are invariant under hereditary head reduction,
which follows directly from the definition:

\begin{lem}\label{lem:BTL}
  Let $M\in \Lambda_{\oplus}$. Then $ \BT M = \BT{L(M)}$.
\end{lem}

It will be sufficient to follow this strategy in order to establish Step~\ref{step:taysup},
\emph{i.e.} the qualitative version of the commutation between normalization
and the Taylor expansion of $\lambda_\oplus$-terms, to be defined
in the next section.

\section{Taylor expansion in a uniform nondeterministic setting}%
\label{sect:resource}

In order to define Taylor expansion, we need to introduce an auxiliary language:
the resource calculus.

\subsection{Resource terms}%
\label{sect:resourceTerms}

	We call \definitive{resource expressions} the elements of
	$\resTerms\cup\resMonomials$,
	where the set $\resTerms$ of \definitive{resource terms}
	and the set $\resMonomials$ of \definitive{resource monomials}
	are defined by mutual induction  as follows:\footnote{
		Recall that the cartesian product of vector spaces
		is given by the disjoint union of bases:
		this is the intuition behind the operators $\linj{-}$ and $\rinj{-}$,
		which will serve in the Taylor expansion of the operator $\oplus$
		of $\Lambda_\oplus$.
		Again, we leave the exact computational behavior of $\oplus$ unspecified,
		and we treat it generically as a pairing operator (without projections):
    in this we follow Tsukada \emph{et al.}~\cite{tao:gen}.
	}
\[
	\resTerms \ni s,t,u,v \recdef
		x \mid \lambda x. s \mid \rappl{s}{\bar{t}}
		\mid \linj{s} \mid \rinj{s}
	\qquad\qquad
	\resMonomials \ni \bar s,\bar t,\bar u,\bar v \recdef
		[s_1,\dotsc,s_n]
\]
and,
in addition to $\alpha$-equivalence, we consider resource expressions
up to permutations of terms in monomials,
so that $[s_1,\dotsc,s_n]$ denotes a multiset of terms.
We give precedence to application and abstraction over
$\linj{-}$ and $\rinj{-}$,
and we write $\rappl s{\bar t_1\cdots\bar t_n}$
for $\rappl{\cdots \rappl s{\bar t_1}\cdots}{\bar t_n}$,
so that we may write $\linj{\lambda x.\rappl s{\bar t\,\bar u}}$
for $\linj{(\lambda x.(\rappl{\rappl s{\bar t}}{\bar u}))}$.
We write $\lambda \vec x.s$ for a term of the form $\lambda x_1.\cdots\lambda x_n.s$.
We moreover write $\bar s\cdot\bar t$ for the multiset union of $\bar s$ and $\bar t$,
and if $\bar s=[s_1,\dotsc,s_n]$ then we write $\length{\bar s}\eqdef n$
for the size of $\bar s$;
in particular $\length{\bar s}=0$ iff $\bar s$ is the empty multiset $[]$,
which is neutral for multiset union.

If $X$ is a set, we write $\sums X$ for the freely generated commutative monoid over $X$:
formally, this is the same as the set of finite multisets of elements of $X$
but we choose to consider its elements as finite linear combinations of elements of $X$ with coefficients in $\N$.
In the following, we write $\resExprs$ for either $\resTerms$
or $\resMonomials$, so that $\resExprSums$ is either $\resTermSums$
or $\resMonomialSums$:
when we consider a sum $E$ of resource expressions,
we always require $E$ to be a sum of terms or a sum of monomials,
\emph{i.e.} $E\in\resExprSums$.
Then we write $\supp(E)\subseteq\resExprs$ for the \definitive{support set} of $E$, which is finite.
We extend  the syntactic constructs of the resource calculus to finite sums of resource expressions by linearity, so that:
\begin{itemize}
	\item if $ S = \sum_{i=1}^{n} s_{i} $ then
		$ \lambda x. S = \sum_{i=1}^{n} \lambda x. s_{i}$,
		$ \bullet \oplus S = \sum_{i=1}^{n}  \bullet \oplus s_{i}$
		and $  S \oplus \bullet = \sum_{i=1}^{n}  s_{i} \oplus \bullet$;
	\item if moreover $ \bar{T} = \sum_{j=1}^{m} \bar{t}_{j} $ then
		$\langle S \rangle \bar{T} = \sum_{i=1}^{n} \sum_{j=1}^{m} \langle s_{i} \rangle \bar{t}_{j}$
		and $ [S] \cdot \bar{T} = \sum_{i=1}^{n} \sum_{j=1}^{m} [s_{i}] \cdot \bar{t}_{j}  $.
\end{itemize}

\noindent
For any resource expression $e\in\resExprs$, we write $n_x(e)$ for the number of occurrences of variable $x$ in $e$.
  If moreover $ \bar{u}= [u_{1},\dotsc,u_{n}] \in \resMonomials $,
	we introduce the \definitive{symmetric $n$-linear substitution} $ \nsubst ex{\bar u}\in\resExprSums$
	of $ \bar{u} $ for the variable $ x $ in $ e $, which is informally defined as follows:
	\[
		\nsubst ex{\bar u} \eqdef \begin{cases}
			\sum\limits_{\sigma\in \permutations{n}} e[ u_{\sigma(1)}/x_{1},\dotsc, u_{\sigma(n)}/x_{n} ] & \text{ if } n_{x}(e) = n
			\\
			0 & \text{ otherwise}
		 \end{cases}
	\]
	where $x_1,\dotsc,x_{n_x(e)}$ enumerate the occurrences of $x$ in $e$.\footnote{
		Enumerating the occurrences of $x$ in $e$ only makes sense
		if we fix an ordering of each monomial in $e$:
		the rigid resource calculus to be introduced later in the paper
		will allow us to give a more formal account of this intuitive presentation.
		For now we stick to the alternative definition given in the next paragraph.
	}

Formally,
$\nsubst ex{\bar u}$ is defined by induction on $e$, setting:
\begin{align*}
	\partial_{x} y \cdot \bar{u}
	&\eqdef
	\begin{cases}   y &  \text{if} \  y \neq x   \text{ and }  n = 0\\
	u_{1} & \text{if }   y = x  \text{ and }  n= 1  \\ 0 & \text{otherwise}\end{cases}
	\\
	\partial_{x} {\lambda y. s}\cdot \bar{u}
	&\eqdef
	\lambda y. ( \partial_{x} s \cdot \bar{u})
	\\
	\partial_{x} {(\linj s)}\cdot \bar{u}
	&\eqdef
	\linj{\partial_{x} s \cdot \bar{u}}
	\\
	\partial_{x} {(\rinj s)}\cdot \bar{u}
	&\eqdef
	\rinj{\partial_{x} s \cdot \bar{u}}
	\\
	\nsubst {\langle s \rangle \bar{t} } {x}  {\bar{u}} &\eqdef
  \sum_{(I_{1}, I_{2}) \in\wcomp{2}{n}} \rappl { \nsubst  {s}{x} {\bar{u}_{I_{1}}} } {\nsubst  {\bar{t}} {x} {\bar{u}_{I_{2}} }}
	\\
	\partial_{x} [t_{1},\dotsc, t_{k}]\cdot \bar{u} &\eqdef
  \sum_{(I_{1},\dotsc, I_{k}) \in\wcomp{k}{n}} [\partial_{x} t_{1} \cdot \bar{u}_{I_{1}},\dotsc, \partial_{x} t_{n} \cdot \bar{u}_{I_{k}} ]
\end{align*}
where $\wcomp{k}{n}$ denotes the set of $k$-tuples $(I_1,\dots,I_k)$
of (possibly empty) pairwise disjoint subsets of $\{1,\dotsc,n\}$
such that $\bigcup_{j=1}^k I_j=\{1,\dotsc,n\}$,\footnote{
	Note that this data is equivalent to a function $\{1,\dotsc,n\}\to\{1,\dotsc,k\}$.
}
and we write $\bar{u}_{\{i_1,\dotsc,i_j\}}\eqdef [u_{i_1},\dotsc,u_{i_j}]$
whenever $1\le i_1<\dotsc<i_j\le n$.
It is easy to check that $\nsubst ex{\bar t}\not=0$ iff $n_x(e)=\length{\bar t}$.

\begin{figure}[t]
	\begin{gather*}
		\begin{prooftree}
			\infer0{\rappl{\lambda x.s}{\bar t}\to_\partial \nsubst sx{\bar t}}
		\end{prooftree}
		\qquad
		\begin{prooftree}
			\infer0{\rappl{\linj s}{\bar t}\to_\partial \linj{\rappl s{\bar t}}}
		\end{prooftree}
		\qquad
		\begin{prooftree}
			\infer0{\rappl{\rinj s}{\bar t}\to_\partial \rinj{\rappl s{\bar t}}}
		\end{prooftree}
		\\[1ex]
		\begin{prooftree}
			\infer0{\lambda x.{(\linj s)}\to_\partial \linj{\lambda x.s}}
		\end{prooftree}
		\qquad
		\begin{prooftree}
			\infer0{\lambda x.{(\rinj s)}\to_\partial \rinj{\lambda x.s}}
		\end{prooftree}
		\\[1ex]
		\begin{prooftree}
			\hypo{s\to_\partial S'}
			\infer1{\lambda x.s\to_\partial\lambda x.S'}
		\end{prooftree}
		\qquad
		\begin{prooftree}
			\hypo{s\to_\partial S'}
			\infer1{\rappl s{\bar t}\to_\partial\rappl{S'}{\bar t}}
		\end{prooftree}
		\qquad
		\begin{prooftree}
			\hypo{\bar s\to_\partial \bar S'}
			\infer1{\rappl{t}{\bar s}\to_\partial \rappl{t}{\bar S'}}
		\end{prooftree}
		\\[1ex]
		\begin{prooftree}
			\hypo{s\to_\partial S'}
			\infer1{\linj s\to_\partial \linj{S'}}
		\end{prooftree}
		\qquad
		\begin{prooftree}
			\hypo{s\to_\partial S'}
			\infer1{\rinj s\to_\partial \rinj{S'}}
		\end{prooftree}
		\qquad
		\begin{prooftree}
			\hypo{s\to_\partial S'}
			\infer1{[s]\cdot{\bar t}\to_\partial[S']\cdot{\bar t}}
		\end{prooftree}
	\end{gather*}
	\caption{Reduction rules of the resource calculus with sums}%
	\label{fig:resred}
\end{figure}

The reduction of the resource calculus is the relation from
resource expressions to finite formal sums of resource expressions
induced by the rules of Figure~\ref{fig:resred}:
the first rule is the counterpart of $\beta$-reduction in the resource calculus;
the next four rules implement the commutation of $\oplus$ with
abstraction and application to a monomial; the final six rules
ensure the contextuality of the resulting relation.
It is extended
to a binary relation on $\sums{\resExprs}$ by setting
$e+F\to_\partial E'+F$ whenever $e\to_\partial E'$.
As in the case of the original resource calculus~\cite{er:tay},
the reduction relation $\to_{\partial}$ is confluent and strongly normalizing.
Confluence may be proved following the same technique
as for the original resource calculus~\cite[Section 3.4]{vaux:taylornf}:
we do not provide any detail, because we will soon focus on a reduction
strategy, which is functional.
For strong normalization, slightly more care is needed, because
the size of expressions does not necessarily decrease under reduction:
\begin{lem}%
  \label{lem:SN}
  The reduction
  $\mathord\to_\partial\subseteq \sums{\resExprs}\times\sums{\resExprs}$
  is strictly normalizing.
\end{lem}
\begin{proof}
  If $e\in\resExprs$, we write $n_\lambda(e)\in\N$
  (resp. $n_\oplus(e)\in\N$)
  for the number of abstractions (resp.\ of $\oplus$) occurring in $e$.
  Let $\#_\oplus(e)$ denote the multiset of natural numbers
  containing a value $n_\oplus (s)$ for each occurrence
  of a subterm $\lambda x.s$ or $\rappl s{\bar t}$ in $e$.
  Formally:
  \begin{gather*}
    \#_\oplus(x) \eqdef []
    \quad
    \#_\oplus(\lambda x.s) \eqdef [n_\oplus(s)]\cdot\#_\oplus(s)
    \quad
    \#_\oplus(\linj s)\eqdef \#_\oplus(s)
    \quad
    \#_\oplus(\rinj s)\eqdef \#_\oplus(s)
    \\
    \#_\oplus(\rappl s{\bar t})\eqdef [n_\oplus(s)]\cdot\#_\oplus(s)\cdot\#_\oplus(\bar t)
    \qquad
    \#_\oplus([s_1,\dotsc,s_n])\eqdef \#_\oplus(s_1)\cdot\cdots\cdot\#_\oplus(s_n)
  \end{gather*}
  where we use the same notations for multisets as for monomials.

  We first establish that, for all $e\in\resExprs$,
  $\bar t\in\resMonomials$ and $e'\in\supp(\nsubst ex{\bar t})$,
  we have $n_\lambda(e')=n_\lambda(e)-1$:
  the proof is by a straightforward induction on $e$.

  Then, whenever $e\to_\partial E'$ and $e'\in\supp(E')$,
  we have $n_\oplus(e')=n_\oplus(e)$, and:
  \begin{enumerate}
    \item\label{snproof:i}
      either $n_\lambda(e')=n_\lambda(e)-1$;
    \item\label{snproof:ii}
      or $n_\lambda(e')=n_\lambda(e)$,
      and we can write $\#_\oplus(e)=\bar n\cdot [n+1]$
      and $\#_\oplus(e')=\bar n\cdot [n]$.
  \end{enumerate}
  The proof is by induction on the derivation of $e\to_\partial E'$:
  the $\beta$-redex case holds using the previous result on multilinear
  substitution to obtain (\ref{snproof:i});
  the other four base cases yield (\ref{snproof:ii});
  and each other case follows straightforwardly by the induction hypothesis.

  Now, for each $E=e_1+\cdots+e_n\in\resExprSums$,
  we write $\#_\lambda(E)=[n_\lambda(e_1),\dotsc,n_\lambda(e_n)]$
  and $\#_\oplus(E)=\#_\oplus(e_1)\cdot\cdots\cdot\#_\oplus(e_n)$.
  By the previous result and the definition of $\to_\partial$ on
  sums of resource expressions:
  if $E\to_\partial E'$ then either $\#_\lambda(E)<\#_\lambda(E)$, or
  $\#_\lambda(E)=\#_\lambda(E)$ and $\#_\oplus(E)<\#_\oplus(E)$,
  considering the multiset order.
  We conclude since the latter is well-founded.
\end{proof}

We write $\NF{E}$ for the unique normal form of $E\in\resExprSums$,
which is a linear operator:
$\NF{\sum_{i=1}^k e_i}=\sum_{i=1}^k \NF{e_i}$.
As stated before,
we do not focus on the reduction relation itself, and we rather
consider the \definitive{hereditary head reduction} strategy obtained by
defining the function $L:\resExprs\to\sums{\resExprs}$
inductively as follows:
\begin{gather*}
	\begin{aligned}
		L(\linj s)
		&\eqdef
		\linj{L(s)}
    &
    &\quad&
		L(\rinj s)
		&\eqdef
		\rinj{L(s)}
		\\
		L( \lambda\vec x.\lambda y. (\linj s))
		&\eqdef
		\lambda\vec x.(\linj{\lambda y.s})
    &
    &\quad&
		L( \lambda\vec x.\lambda y. (\rinj s))
		&\eqdef
		\lambda\vec x.(\rinj{\lambda y.s})
	\end{aligned}
	\\
	\begin{aligned}
		L(\lambda\vec x.\rappl{\rappl{\linj {s}}{\bar t}}{\bar u_1\cdots\bar u_k})
		&\eqdef
		\lambda\vec x.\rappl{\linj {\rappl{s}{\bar t}}}{\bar u_1\cdots\bar u_k}
		\\
		L(\lambda\vec x.\rappl{\rappl{\rinj {s}}{\bar t}}{\bar u_1\cdots\bar u_k})
		&\eqdef
		\lambda\vec x.\rappl{\rinj {\rappl{s}{\bar t}}}{\bar u_1\cdots\bar u_k}
		\\
		L(\lambda\vec x.\rappl{y}{\bar{s}_{1}\cdots \bar{s}_{k}})
		&\eqdef
		\lambda\vec x.\rappl{y}{L(\bar{s}_{1})\cdots L(\bar{s}_{k})}
		\\
		L([s_{1}, \dots , s_{k}])
		&\eqdef
		[L(s_{1}), \dots ,L(s_{k})  ]
		\\
		L(\lambda\vec x. \rappl{\lambda y. s}{\bar t\,{\bar u_1\cdots\bar u_k}})
		&\eqdef
		\lambda\vec x.\rappl{\nsubst s y{\bar t}}{{\bar u_1\cdots\bar u_k}}
	\end{aligned}
\end{gather*}
extended to sums of resource expressions by linearity,
setting $L(\sum_{i=1}^k e_i)\eqdef\sum_{i=1}^k L(e_i)$.

Again, it should be clear that $e\to_\partial^* L(e)$:
if $e$ contains a redex (\emph{i.e.} the left-hand side of any of the first
five rules of Figure~\ref{fig:resred}) in head position,
then $L(e)$ is obtained by firing this redex;
otherwise each term in a monomial argument of the head variable is reduced,
following the same strategy inductively.
Moreover, $e=L(e)$ iff $e$ is normal:
we obtain an equivalence because $\to_\partial$ is strongly normalizing
on sums of resource expressions and, if $e$ is not normal,
$L(e)$ is obtained by firing at least one redex in $e$.

Due to the definition of $\to_\partial$ on sums and the linearity of $L$,
these properties extend directly: $E\to_\partial^* L(E)$ and
$E=L(E)$ iff $E$ is normal (\emph{i.e.} it is a sum of normal expressions).
It moreover follows that $L$ is normalizing:
for all $E\in\resExprSums$, there is $n$ such that $L^n(E)=\NF E$.

\subsection{Taylor expansion of \texorpdfstring{$\lambda_{\oplus}$}{lambdaplus}-terms}%
\label{subsection:taylor}

The Taylor expansion of a $\lambda_\oplus$-term will be an infinite linear combination
of resource terms: to introduce it, we first need some preliminary notations
and results.

If $X$ is a set, we write $\vectors{X}$ for the set of
possibly infinite linear combinations of elements of $X$
with non negative rational coefficients (in fact we could use any commutative semifield):
equivalently, $\vectors{X}$ is the set of functions from $X$ to
the set of non negative rational numbers.
We write $A=\sum_{a\in X} A_a\,a\in\vectors{X}$ and then
the \definitive{support set} of $A$ is $\supp(A)=\{a\in X\mid A_a\not=0\}$.

All the syntactic constructs of resource expressions are extended
to infinite linear combinations, componentwise:
\begin{itemize}
	\item if $ S\in\resTermVectors$ then
    \[
      \lambda x.S \eqdef \sum_{s\in\resTerms}S_s(\lambda x.s)
      \quad,\quad
      \linj S \eqdef \sum_{s\in\resTerms}S_s (\linj s)
      \quad \text{and}\quad
      \rinj S \eqdef \sum_{s\in\resTerms}S_s (\rinj s)
      \quad;
    \]
	\item if moreover $ \bar{T}\in\resMonomialVectors$ then
    \[
      \langle S \rangle \bar{T} \eqdef \sum_{s\in\resTerms} \sum_{\bar t\in\resMonomials} S_s \bar T_{\bar t} (\rappl s{\bar{t}})
      \quad;
    \]
  \item and if $S_1,\dotsc,S_n\in\resTermVectors$ then
    \[
      [S_1,\dotsc,S_n] \eqdef \sum_{(s_1,\dotsc,s_n)\in\resTerms^n} \big(\prod_{i=1}^n {S_i}_{s_i}\big) [s_1,\dotsc,s_n]
      \quad.
    \]
\end{itemize}
Observe indeed that each of these infinite sums is finite in each component:
e.g., for each $\bar s\in\resMonomials$,
there are finitely many tuples $(s_1,\dotsc,s_n)\in\resTerms^n$
such that $\bar s = [s_1,\dotsc,s_n]$.

Similarly we extend syntactic constructs to sets of resource expressions:
\begin{itemize}
	\item if $ S\subseteq\resTerms$ then
    \[
      \lambda x.S \eqdef \{ \lambda x.s \mid s\in S\}
      \quad,\quad
      \linj S \eqdef \{ \linj s \mid s\in S\}
      \quad \text{and}\quad
      \rinj S \eqdef \{ \rinj s \mid s\in S\}
      \quad;
    \]
	\item if moreover $ \bar{T}\subseteq\resMonomials$ then
    \[
      \langle S \rangle \bar{T} \eqdef \{\rappl s{\bar{t}}\mid s\in S,\ \bar t\in\bar T\}
      \quad;
    \]
  \item and if $S_1,\dotsc,S_n\subseteq\resTerms$ then
    \[
      [S_1,\dotsc,S_n] \eqdef \{  [s_1,\dotsc,s_n] \mid s_i\in S_i\text{ for }1\le i\le n\}
      \quad.
    \]
\end{itemize}
Considering subsets of $\resExprs$ as infinite
linear combinations of resource expressions with boolean coefficients,
this is just a variant of the previous construction
(which can be carried out in any commutative semifield).
Moreover, syntactic constructs commute with the support function: e.g.,
$\lambda x.\supp(S)=\supp(\lambda x.S)$.

Let $S \in \resTermVectors$. We define $S^{n}\in\resMonomialVectors$ by induction on $n$:
$S^{0}= []$ and $S^{n+1} = [S]\cdot S^{n}$.
Then we define the \definitive{promotion} of $S$ as the series $S^\oc = \sum_{n=0}^{\infty} \frac{1}{n!} S^{n}$:
because the supports of $S^n$ and $S^p$ are disjoint when $n\not=p$,
this sum is componentwise finite.
If $S\subseteq\resTerms$ is a set of terms,
we may also write $S^\oc=\{[s_1,\dotsc,s_n]\mid s_1,\dotsc,s_n\in S\}$
for the set of monomials of terms in $S$,
so that $\supp(S^\oc)=\supp(S)^\oc$ for any $S\in\resTermVectors$.

We define the \definitive{Taylor expansion} $\tayexp{M} \in \resTermVectors$
of $M\in\ndTerms$ inductively as follows:
\begin{align*}
\tayexp{x}
&\eqdef x
\\
\tayexp{(\lambda x.N)}
&\eqdef \lambda x. \tayexp{N}
\\
\tayexp{(PQ)}
&\eqdef  \rappl{\tayexp{P}}{(\tayexp{Q})^{!}}
\\
\tayexp{(P\oplus Q)}
&\eqdef (\linj{\tayexp{P}})+(\rinj{\tayexp{Q}})
\quad .
\end{align*}
Note that this definition follows the one for the ordinary
$\lambda$-calculus given by Ehrhard and Regnier~\cite{er:tay},
in the form described in their Lemma~18.
We extend it to $\oplus$ by encoding the pair of vectors
$(\tayexp P,\tayexp Q)$ as the sum vector $(\linj{\tayexp{P}})+(\rinj{\tayexp{Q}})$.\footnote{
  Note that the original notion of Taylor expansion for nondeterministic $\lambda$-terms
  (considered as algebraic $\lambda$-terms without coefficients)
  interprets nondeterministic choice directly as a sum,
  setting $\tayexp{(M\oplus N)}=\tayexp M+\tayexp N$~\cite{ehrhard:finres,ptv:taylorsn,vaux:taylornf}.
  Following Tsukada, Asada and Ong~\cite{tao:gen}, we can recover this notion,
  by erasing the markers $\linj{\mathord -}$ and $\rinj{\mathord -}$,
  with one \emph{caveat}:
  in general, this might yield infinite sums of coefficients,
  because a single resource term without markers
  may be obtained from infinitely many terms with markers.
  Define for instance $x\oplus n\bullet$ by analogy with Example~\ref{exa:BT}:
  $x\oplus 0\bullet=x$ and $x\oplus(n+1)\bullet=(x\oplus n\bullet)\oplus\bullet$.
  Then forgetting markers in the sum $\sum_{i=0}^\infty x\oplus n\bullet$
  yields $\sum_{i=0}^\infty x$.
  And it turns out that normalizing the Taylor expansion
  of nondeterministic terms does yield such sums:
  see Example~\ref{exa:TBT}.
}

\begin{exa}
  We have $\tayexp{(x\oplus x)}=(\linj x)+(\rinj x)$
  hence \[
    \tayexp{(\lambda x.(x\oplus x))}
    =\lambda x.((\linj x)+(\rinj x))
    =(\lambda x.(\linj x))+(\lambda x.(\linj x))
  \] and \[
    \tayexp{(y(x\oplus x))}
    = \sum_{n\in\N}\frac 1{n!}\rappl y{[(\linj x)+(\rinj x)]^n}
    = \sum_{n\in\N}\sum_{i=0}^n\frac 1{i!(n-i)!}
    \rappl y{[\linj x]^i\cdot[\rinj x]^{n-i}}
    \quad.
  \]
\end{exa}

Writing $\taysup M\eqdef \supp(\tayexp M)$ for the support of
Taylor expansion, we obtain:
\begin{align*}
	\taysup{x} & = \{x\}
	\\
	\taysup{\lambda x.N} & = \lambda x. \taysup{N} = \{\lambda x.t\mid t\in\taysup N\}
	\\
	\taysup{PQ} &=  \rappl{\taysup{P}}{\taysup{Q}^{!}} =
	\{ \rappl s{[t_1,\dotsc,t_n]} \mid s\in\taysup{P}\text{ and }t_1,\dotsc,t_n\in\taysup{Q} \}
	\\
	\taysup{P\oplus Q} & = (\linj{\taysup{P}})\cup (\rinj{\taysup{Q}}) =
	\{ \linj{s} \mid s\in\taysup P\}\cup \{ \rinj{t} \mid t\in\taysup Q\}
\end{align*}
so that $\tayexp M=\sum_{s\in\taysup M} \tayexp M_s s$.

We can immediately check that Step~\ref{step:clique} still holds for our extension
of Taylor expansion to $\lambda_\oplus$-terms:
we prove that $\taysup M$ is always a clique for the coherence relation
$\mathord\coh \subseteq \resExprs \times \resExprs$
inductively defined by the rules of Figure~\ref{fig:coh}.
\begin{figure}[t]
	\begin{gather*}
		\begin{prooftree}
			\infer0[]{x \coh x}
		\end{prooftree}
		\qquad
		\begin{prooftree}
			\hypo{ s \coh s'}
			\infer1[]{ \lambda x. s \coh \lambda x. s' }
		\end{prooftree}
		\qquad
		\begin{prooftree}
			\hypo{ s \coh s'}
			\hypo{ \bar{t} \coh \bar{t}'}
			\infer2[]{ \langle s \rangle \bar{t} \coh \langle s' \rangle \bar{t}'}
		\end{prooftree}
		\qquad
		\begin{prooftree}
			\hypo{t_{i} \coh t_{j} \text{ for }1\le i,j\le n + m}
			\infer1[]{ [t_{1},\dotsc, t_{n}] \coh [t_{n + 1},\dotsc, t_{n + m}] }
		\end{prooftree}
		\\[1ex]
		\begin{prooftree}
			\hypo{ s \coh s'}
			\infer1[]{ s \oplus \bullet \coh s' \oplus \bullet }
		\end{prooftree}
		\qquad
		\begin{prooftree}
			\hypo{ s \coh s'}
			\infer1[]{ \bullet \oplus s \coh \bullet \oplus s' }
		\end{prooftree}
		\qquad
		\begin{prooftree}
			\infer0[]{s \oplus \bullet \coh \bullet \oplus s'}
		\end{prooftree}\quad
	\end{gather*}
	\caption{Rules for the coherence relation on $\resExprs$.}%
	\label{fig:coh}
\end{figure}
The first four rules are exactly those
for the ordinary resource calculus~\cite[Section 3]{er:tay},
while the last three rules are reminiscent of the definition
of the cartesian product of coherence spaces~\cite[Definition~5]{girard:ll}.
Again, this is consistent with the fact that we treat $\oplus$ as a pairing
construct, denoting an unspecified superposition operation.

Observe that the relation $\coh$ is automatically symmetric,
but not reflexive: e.g., $[s,t]\not\coh[s,t]$ when $s\not\coh t$.
We say a resource expression $e$ is \definitive{uniform} if $e\coh e$,
so that uniform expressions form a coherence space in the usual sense.\footnote{
  Note that, by contrast with the coherence relation considered by Dal Lago and
  Leventis for the Taylor expansion of probabilistic $\lambda$-terms~\cite{dll:teplt},
  $e\coh e'$ does not imply the uniformity of $e$ nor $e'$:
  we have $\linj s\coh\rinj {s'}$ without any condition on $s$ and $s'$.
  We could adapt our main results with a finer coherence, similar to theirs,
  requiring $s\coh s$ and $s'\coh s'$ for $\linj s\coh\rinj {s'}$ to hold:
  uniform expressions and cliques are the same for both relations.
  Nonetheless, we find it interesting that this additional hypothesis is not
  needed for Step~\ref{step:disjoint}.
}
We call \definitive{clique} any set $E$ of resource expressions
such that $e\coh e'$ for all $e,e'\in E$.
In particular, the elements of a clique are necessarily uniform.

We obtain the expected result
by a straightforward induction on $\lambda_\oplus$-terms:
\begin{thm}[Step~\ref{step:clique}]%
	\label{thm:clique}
	The Taylor support $\taysup{M}$ is a clique.
\end{thm}

\subsection{Multiplicity coefficients}
We now generalize Step~\ref{step:m} in our generic nondeterministic setting:
we can define a multiplicity coefficient $m(s)\in\N$ for each $s\in\resTerms$ so
that $\tayexp M_s=\frac 1{m(s)}$ whenever $s\in\taysup M$.

Given any set $X$ and $n\in\N$,
we consider the left action of the group $\permutations{n}$
of all permutations of $\{1,\dotsc,n\}$
on the set $X^n$ of $n$-tuples, defined as follows:
if $\vec a=(a_1,\dotsc,a_n)$ and $\sigma\in\permutations{n}$ then
$\lact{\sigma}{\vec a}=(a_{\sigma^{-1}(1)},\dotsc,a_{\sigma^{-1}(n)})$.
Writing $\lact{\sigma}{\vec a}=(a'_1,\dotsc,a'_n)$,
we obtain $a'_{\sigma(i)}=a_i$.
Let us recall that if $\vec a\in X^n$, then the stabilizer of $\vec a$
is $\St(\vec a)=\{\sigma\in\permutations{n} \mid \lact{\sigma}{\vec a}=\vec a\}$.

If $\vec s=(s_1,\dotsc,s_n)\in \resTerms^n$ and $S\in\resTermVectors$,
we write $S^{\vec s}=\prod_{i=1}^n S_{s_i}$:
observe that this does not depend on the ordering of the $s_i$'s,
so if $\bar s=[s_1,\dotsc,s_n]\in\resMonomials$,
we may as well write $S^{\bar s}=S^{(s_1,\dotsc,s_n)}$.
We obtain:

\begin{lem}\label{lem:coeffProm}
Let $S\in \resTermVectors$ and $\bar{s}\in \supp(S^{!})$.
If $\vec{s}= (s_{1},\dotsc, s_{n})$ is an enumeration of $\bar{s}$,
\emph{i.e.} $ [s_{1},\dotsc, s_{n}] = \bar{s}$,
then $ (S^{!})_{\bar{s}} = \dfrac{S^{\bar{s}}}{\card(\St(\vec{s}))}$.
\end{lem}

\begin{proof}
  By the definition of promotion, and by linearity, we obtain
	\[S^{!}= \sum_{n=0}^{\infty}
    \frac 1{n!}\sum_{(s_{1},\dotsc, s_{n})\in\resTerms^n}
    S^{(s_{1},\dotsc, s_{n})}[s_1,\dotsc,s_n]\quad.\]
	If $\length{\bar{s}}=n$, we thus obtain:
	\[(S^{!})_{\bar s}=\card(\{(s_1,\dotsc,s_n)\mid[s_1,\dotsc,s_n]=\bar s\}) \frac{S^{\bar s}}{n!}\quad.\]
	Observing that $\{(s_1,\dotsc,s_n)\mid[s_1,\dotsc,s_n]=\bar s\}$
	is the orbit of any enumeration of $\bar s$
	under the action of $\permutations{n}$,
	and that $\card(\permutations{n})=n!$, we conclude by Fact~\ref{fact:bijG}. 
\end{proof}

Let $s\in \resTerms$. We inductively define $m(s)$, the \definitive{multiplicity coefficient}
of $s$, as follows:
\begin{align*}
m(x)
&\eqdef 1
\\
\left.
\begin{array}{r}
m(\lambda x. s)\\ m( \linj s) \\ m(\rinj s)
\end{array}
\right\}
&\eqdef m(s)
\\
m(\langle s \rangle \bar{t})
& \eqdef m(s)m(\bar{t})
\\
m([t_{1}]^{n_{1}}\cdot \dots \cdot [t_{n}]^{n_{n}})
&\eqdef \prod_{i=1}^{n} n_{i} ! \ m(t_{i})^{n_{i}} 
\end{align*}
assuming the $t_i$'s are pairwise distinct in the case of a monomial.
Again, this definition extends straightforwardly
the one given by Ehrhard and Regnier for
their resource calculus~\cite[Section 2.2.1]{er:tay},
given that $\linj{\mathord -}$ and $\rinj{\mathord -}$ are both linear.
Observe that, considering the function $m$ as a vector
$m\in \vectors{\resExprs}$,
if $\vec s$ is an enumeration of $\bar s$ then
$m(\bar s)=m^{\bar s}\card(\St(\vec s))$.

\begin{thm}[{Step~\ref{step:m}}]\label{thm:m}
Let $s\in \taysup{M}$. Then $  \tayexp{M}_{s}= \dfrac{1}{m(s)}$.
\end{thm}
\begin{proof}
	The only interesting case is that of an application: $M=PQ$.
	Assume $s\in\taysup{M}$;
	then $s=\rappl u{\bar v}$
	with $u\in\taysup P$ and $\bar v=[v_1,\dotsc,v_n]\in\taysup Q^\oc$.
	By definition,
	$\tayexp M_s
	=(\rappl{\tayexp P}{(\tayexp Q)^\oc})_{\rappl u{\bar v}}
	=\tayexp P_u (\tayexp Q)^\oc_{\bar v}
	$.
	Setting $\vec v=(v_1,\dotsc,v_n)$,
	we obtain $\tayexp M_s={\tayexp P_u(\tayexp Q)^{\bar v}}/{\card(\St(\vec v))}$
	by Lemma~\ref{lem:coeffProm}.
	By the induction hypothesis applied to $P$ and $Q$,
	we obtain $1/\tayexp P_u=m(u)$ and $1/\tayexp Q_{v_i}=m(v_i)$
	hence $1/\tayexp M_s=m(u)m^{\bar v}\card(\St(\vec v))=m(u)m(\bar v)=m(s)$.
\end{proof}

We can as well obtain Step~\ref{step:coefNF} following Ehrhard and Regnier's
study of permutations of variables occurrences, but here we choose to depart
from their approach.
At this point, indeed, we hope the reader will share our opinion that
the combinatorics of Taylor expansion is more intimately connected with the
action of permutations on the enumerations of monomials occurring in resource
expressions.

In the upcoming Section~\ref{sect:rigid}, we propose to flesh out this viewpoint, and to recast
resource expressions as equivalence classes of their rigid (\emph{i.e.}
non-commutative) representatives, up to the isomorphisms of a groupoid of
permutation terms inductively defined on the syntactic structure.

The other remaining Steps~\ref{step:disjoint} and~\ref{step:taysup} are
purely qualitative properties of the Taylor support.
We choose to treat also Step~\ref{step:disjoint} in the rigid setting,
to be introduced later, because it is essentially a property of rigid
reduction.
On the other hand, the commutation of Step~\ref{step:taysup} can be established
directly.

\subsection{Taylor expansion of Böhm trees}

The Taylor expansion of a B\"{o}hm tree is obtained as follows.
First we extend the definition of Taylor expansion from $\ndTerms$ to
$\ndApprox$ by adding the inductive case $\tayexp\bot\eqdef 0$,
hence $\taysup\bot = \emptyset$.
Then we set $\taysup{\BT{M}} \eqdef \bigcup_{B\in\BT M}\taysup{B}$.

We can already observe that if $B\in\eBTs$ and $s\in\taysup{B}$ then
$s$ is normal: indeed, the absence of redexes
is preserved by the inductive definition of Taylor expansion.
It follows that any $s\in\taysup{\BT{M}}$ is normal.
Moreover, it is clear that Theorem~\ref{thm:m} extends
to term approximants, hence
$\tayexp B_s=\frac 1{m(s)}$ whenever $s\in \taysup B$.
Thus, it makes sense to define the Taylor expansion
of a Böhm tree as:
$\tayexp{\BT{M}} \eqdef \sum_{s\in\taysup{\BT M}}\frac 1{m(s)} s$.

\begin{exa}%
  \label{exa:TBT}
  Recall from Example~\ref{exa:BT} that if we set
  $M=\Theta\lambda y.(y\oplus x)$ then
  $\BT{M}=\{\bot\oplus nx\mid n\in\N \}$.
  Observe that $\taysup{\bot\oplus nx}=\{(\bullet\oplus x)\oplus i\bullet\mid 0\le i<n\}$
  so that $\taysup{\BT{M}}=\{ (\bullet\oplus x)\oplus n\bullet\mid n\in\N\}$
  and $\tayexp{\BT{M}}=\sum_{i=0}^\infty (\bullet\oplus x)\oplus n\bullet$,
  because $m((\bullet\oplus x)\oplus n\bullet)=1$ for each $n\in\N$.
\end{exa}

We shall achieve Step~\ref{step:taysup} by showing that
the parallel left strategy in $\ndTerms$ can be simulated
in the support of Taylor expansion,
and that $\taysup{\BT M}$ is formed by accumulating
the normal forms reached from $\taysup M$ by this strategy.

First, we extend the operations $\nsubst -x-$, $L(-)$ and $\NF -$
to sets of resource expressions in the following way:
\[
  \nsubst Ex{\bar T}
  \eqdef
  \bigcup_{e\in E}\bigcup_{\bar t\in \bar T}\supp(\nsubst ex{\bar t})
  \ ,\quad
  L (E)
  \eqdef
  \bigcup_{e\in E} \supp(L(e))
  \quad\text{and}\quad
  \NF E
  \eqdef
  \bigcup_{e\in E} \supp(\NF e)
\]
whenever $E\subseteq\resExprs$ and $\bar T\subseteq\resTerms$.\footnote{
  In contrast with the case of syntactic constructors in Section~\ref{subsection:taylor},
  extending these operations to infinite linear combinations rather than sets
  requires some work.

  In the case of $\nsubst -x-$,
  we can show that each expression $e'$
  is in the support of finitely many sums of the shape
  $\nsubst ex{\bar t}$, by observing that the size
  of the antecedents $e$ and $\bar t$ is at most that of $e'$~\cite[Lemma 3.7]{vaux:taylornf}.
  Then one can exploit the fact that the redexes
  fired in the reduction from $e$ to $L(e)$
  are pairwise independent,
  to deduce that each $e'$
  is in the support of finitely many sums of the shape
  $L(e)$:
  this is a particular case of a result established
  by the second author for parallel reduction~\cite[Section 6.2]{vaux:taylornf}.

  The case of $\NF -$ is even more intricate because, given an infinite linear
  combination $S$ of resource terms, the sum $\sum_{s\in \resTerms} S_s\NF s$ is
  not well defined in general --- indeed, it is easy to find an infinite family of
  resource terms, all having the same nonzero normal form.
  Uniformity is one solution to this issue:
  if the support of $S$ is a clique then the summands
  $\NF s$ for $s\in \supp(S)$ have pairwise disjoint supports.
  This result is the main ingredient of Step~\ref{step:disjoint}:
  it will be our Theorem~\ref{thm:disjoint} below.
  For a survey of alternative approaches we refer to the study
  of this subject by the second author~\cite{vaux:taylornf}.
}

\begin{lem}\label{lem:commL} Let $M$ be a $\lambda_{\oplus}$-term. Then $ L(\taysup{M}) = \taysup{L(M)}$.
\end{lem}
\begin{proof}
The proof is the same as for $\lambda$-terms~\cite{vaux:taylornf},
the case of $\oplus$ being direct.
The base case requires to prove that
$\taysup{\subst MxN}=\nsubst{\taysup M}x{\taysup N^\oc}$,
which is done by a straightforward induction on $M$.
\end{proof}

\begin{lem}\label{lem:orderBT}
Let $A,B \in \ndApprox$. If $ A\leq B $ then $ T(A) \subseteq T(B)$.
\end{lem}
\begin{proof}
  By straightforward induction on the derivation of $A\leq B$.
\end{proof}


\begin{lem}\label{lem:eBTtaysup}
	For any $M\in\ndTerms$,
	$\taysup{\eBT M} = \{ s\in\taysup M\mid \text{$s$ is normal} \}$.
\end{lem}
\begin{proof}
  We have
	$\taysup{\eBT M} \subseteq \taysup M$
  by Lemma~\ref{lem:orderBT}
	and the obvious fact that $\eBT M\le M$.
  We deduce the inclusion $\subseteq$,
  recalling that the Taylor support of elementary
  Böhm trees contains normal terms only.

	Conversely, if $s\in\taysup M$ and $s$ is normal,
	then either $M=N\oplus P$ and $s=\linj t$ or $s=\rinj u$
	with $t\in\taysup N$ or $u\in\taysup P$;
	or $M=\lambda \vec x.  x Q_{1}\cdots Q_{k}$  and
	$s=\lambda \vec x.\rappl x{\bar q_1\cdots \bar q_k}$
	with $\bar q_i\in\taysup{Q_i}^\oc$ for $1\le i\le k$.
	We obtain inductively
	$t\in\taysup{\eBT N}$
	or
	$u\in\taysup{\eBT P}$
	or $\bar q_i\in\taysup{\eBT{Q_i}}^\oc$ for $1\le i\le k$,
	and then $s\in\taysup{\eBT M}$.
\end{proof}

Step~\ref{step:taysup} then follows, using the fact that $\BT M$ is the
downwards closure of $\{\eBT{L^n(M)}\mid n\in\N\}$:

\begin{thm}[Step~\ref{step:taysup}]\label{thm:taysup}
Let $M\in \ndTerms$. Then $\taysup{\BT{M}}= \NF{\taysup{M}}$.
\end{thm}
\begin{proof}
Recall that $\NF{\taysup{M}} = \bigcup_{s\in \taysup{M}} \supp(\NF{s})$.
The proof is by double inclusion.

$({\subseteq})$  Let $t\in \taysup{\BT{M}}$,
\emph{i.e.} $t\in\taysup B$ for some $B\in\BT M$.
By the definition of $\BT M$, there exists $n\in\N$ such that
$B\le \eBT{L^n(M)}$, and then by Lemma~\ref{lem:orderBT}
$t\in\taysup{\eBT{L^n(M)}}$.
By Lemma~\ref{lem:eBTtaysup},
$t$ is normal and $t\in\taysup{L^n(M)}$.
By Lemma~\ref{lem:commL},
$t\in L^n(\taysup M)$,
hence there exists $s\in\taysup M$ such that
$t\in \supp(L^n(s))$.
Since $t$ is normal, $t\in \supp(\NF s)$.

$({\supseteq}) $ If $t\in \NF{\taysup{M}}$ we can fix $s\in \taysup{M}$
such that $t\in \supp(\NF{s})$.
Then there exists $n\in\N$ such that $\NF s=L^n(s)$,
hence $t\in\bigcup_{s \in \taysup M} \supp(L^n(s)) = L^n(\taysup M)$.
By Lemma~\ref{lem:commL}, $t\in \taysup{L^n(M)}$
and since $t$ is normal, Lemma~\ref{lem:eBTtaysup}
entails that $t\in\taysup{\eBT{L^n(M)}}$.
By the definitions of $\BT M$ and $\taysup{\BT M}$,
we have $\eBT{L^n(M)}\subseteq\BT M$ and then $\taysup{\eBT{L^n(M)}}\subseteq\taysup{\BT M}$,
and we obtain $t\in\taysup{\BT{M}}$.
\end{proof}

\section{The groupoid of permutations of rigid resource terms}%
\label{sect:rigid}

\subsection{Rigid resource terms and permutation terms}

We introduce the set $D$ of \definitive{rigid resource terms}
and the set $D^{!}$ of \definitive{rigid resource monomials}
by mutual induction as follows:
\[
	\rigidTerms \ni a,b,c,d \recdef x \mid \lambda x. a \mid \rappl{a}{\vec{b}} \mid \rinj a \mid \linj a
	\qquad
	\qquad
	\rigidMonomials \ni \vec{a},\vec{b},\vec{c},\vec{d} \recdef (a_1,\dotsc,a_n)
	\quad.
\]
Rigid resource terms are considered up to renaming of bound variables:
the only difference with resource terms is that rigid monomials are ordered
lists rather than finite multisets.
We write $\length{(a_1,\dotsc,a_n)}\eqdef n$,
and $(a_1,\dotsc,a_n)\cons(a_{n+1},\dotsc,a_{n+m}) \eqdef (a_1,\dotsc,a_{n+m})$.
We write $\rigidExprs$ for either $\rigidTerms$ or $\rigidMonomials$
and call \definitive{rigid resource expression} any rigid term
or rigid monomial.
Again, for any $r\in\rigidExprs$, we write $n_x(r)$ for the number of free occurrences of the variable $x$ in $r$,
and we use notations and priority conventions similar to those for non rigid expressions:
e.g., we may write
$\linj{\lambda \vec x.\rappl a{\vec b\,\vec c}}$
for $\linj{(\lambda x_1.\dotsc.\lambda x_n.(\rappl{\rappl a{\vec b}}{\vec c}))}$.

As we have already stated, rigid resource expressions are nothing but resource
expressions for which the order of terms in monomials matter.
To make this connexion formal, consider the
\definitive{representation relation} $\lhd \subseteq D^{(!)}\times \resExprs$
defined by the rules of Figure~\ref{fig:lhd}.
\begin{figure}[t]
	\begin{gather*}
		\begin{prooftree}
			\infer0[]{ x \lhd x  }
		\end{prooftree}
		\qquad
		\begin{prooftree}
			\hypo{  a \lhd s }
			\infer1[]{  \lambda x. a \lhd \lambda x. s }
		\end{prooftree}
		\qquad
		\begin{prooftree}
			\hypo{  a \lhd s }
			\infer1[]{  \linj a \lhd s \oplus \bullet }
		\end{prooftree}
		\qquad
		\begin{prooftree}
			\hypo{  a \lhd s }
			\infer1[]{ \rinj  a \lhd \bullet \oplus s }
		\end{prooftree}
		\\[1ex]
		\begin{prooftree}
			\hypo{  c\lhd s}
			\hypo{ \ \vec{d}\lhd \bar{t}}
			\infer2[]{  \langle c \rangle \vec{d} \lhd \langle s \rangle \bar{t} }
		\end{prooftree}
		\qquad
		\begin{prooftree}
			\hypo{a_{1}\lhd t_{1}\quad\cdots\quad a_{n}\lhd t_{n}}
			\infer1[]{  (a_{1},\dotsc, a_{n}) \lhd [t_{1},\dotsc, t_{n}]   }
		\end{prooftree}
	\end{gather*}
	\caption{Rules for the rigid representation relation}%
	\label{fig:lhd}
\end{figure}
Observe that the relation $\lhd$ is the graph of a surjection $\rigidExprs\to\resExprs$:
if $r\in\rigidExprs$, there exists a unique $e\in\resExprs$ such
that $r\lhd e$, and then we write $\softify r\eqdef e$;
and any $e\in\resExprs$ has at least one rigid representation $r\lhd e$.
Moreover observe that, if $\vec a\lhd \bar t$ and $\length{\vec a}=n$
then for any $\sigma\in \permutations{n}$, $\lact{\sigma}{\vec a}\lhd \bar t$,
\emph{i.e.} $\softify{\lact{\sigma}{\vec a}}=\softify{\vec a}$.

We now introduce a syntax for the trees of permutations that can act on
monomials at any depth in a rigid expression.
The language of such \definitive{permutation expressions} is given as follows:
\[ \D \ni \alpha, \beta, \gamma, \delta
	\recdef id_x \mid \lambda x.\alpha
	\mid \rappl \alpha {\tilde \beta} \mid \linj \alpha \mid \rinj \alpha
	\qquad
	\Dm \ni
	\tilde \alpha,
	\tilde \beta,
	\tilde \gamma,
	\tilde \delta
	\recdef {(\sigma,(\alpha_1,\dotsc,\alpha_n))}
\]
where $x$ ranges over variables and $\sigma$ ranges over $\permutations{n}$
in the pair $(\sigma,(\alpha_1,\dotsc,\alpha_n))$.
In other words, a \definitive{permutation term} (resp. \definitive{permutation
monomial}) is nothing but a rigid term (resp.\ rigid monomial), with a
permutation attached with each list of arguments.
In general, we will simply write  $(\sigma,\alpha_1,\dotsc,\alpha_n)$ for
the permutation monomial $(\sigma,(\alpha_1,\dotsc,\alpha_n))$.

We say $\epsilon\in\G$ maps $r\in\rigidExprs$ to $r'\in\rigidExprs$
if the statement $\epsilon:r\cong r'$ is derivable from the rules of Figure~\ref{fig:cong}.
\begin{figure}[t]
	\begin{gather*}
		\begin{prooftree}
			\infer0{ id_{x} : x \cong x  }
		\end{prooftree}
		\qquad
		\begin{prooftree}
			\hypo{ \alpha : a \cong a' }
			\infer1{ \lambda x. \alpha : \lambda x. a \cong \lambda x. a' }
		\end{prooftree}
		\qquad
		\begin{prooftree}
			\hypo{ \gamma : c\cong c'}
			\hypo{ \delta : \vec{d}\cong \vec{d}'}
			\infer2{ \rappl{\gamma}{\delta}: \langle c \rangle \vec{d} \cong \langle c' \rangle \vec{d}' }
		\end{prooftree}
		\\[1ex]
		\begin{prooftree}
			\hypo{ \alpha : a \cong a' }
			\infer1{ \linj \alpha : \linj a \cong \linj {a'} }
		\end{prooftree}
		\qquad
		\begin{prooftree}
			\hypo{ \alpha : a \cong a' }
			\infer1{ \rinj \alpha : \rinj a \cong \rinj {a'} }
		\end{prooftree}
		\\[1ex]
		\begin{prooftree}
			\hypo{\sigma \in \permutations{n}}
			\hypo{\alpha_{1}:a_{1}\cong a'_{\sigma(1)}\quad\cdots\quad\alpha_{n}:a_{n}\cong a'_{\sigma(n)}}
			\infer2{(\sigma,\alpha_{1},\dotsc,\alpha_{n}):(a_{1},\dotsc,a_{n})\cong(a'_{1},\dotsc,a'_{n})}
		\end{prooftree}
	\end{gather*}
	\caption{Permutation expressions as morphisms between rigid expressions}%
	\label{fig:cong}
\end{figure}
We then write $ r \cong r'$ if there exists some $\epsilon \in \G$
such that $ \epsilon: r \cong r' $.
As a direct consequence of the definitions,
we obtain that $\cong$ is nothing but the equivalence kernel
of the function $r\in\rigidExprs\mapsto\softify r\in\resExprs$:
\begin{lem}\label{lem:conglhd}
	For all $r,r'\in\rigidExprs$,
	$r\cong r'$ iff $\softify r=\softify{r'}$.
\end{lem}
The equivalence classes for $\cong$ are thus exactly the sets
of rigid representations of each resource expression.
We can moreover organize the permutation expressions witnessing
this equivalence relation into a groupoid, whose objects
are resource expressions.
Observe indeed that, for each $\epsilon\in\G$ there is exactly one pair
$(r,r')$ of rigid expressions such that $\epsilon:r\cong r'$.
Given $r,r'\in\rigidExprs$, the set of morphisms from $r$ to $r'$ is then
$ \G(r,r') = \{ \epsilon\mid  \epsilon: r\cong r' \}$.
The composition $\epsilon'\epsilon\in\G(r,r'')$ of $\epsilon\in\G(r,r')$ and $\epsilon'\in\G(r',r'')$
is defined by induction on the syntax of rigid resource expressions in the obvious way:
the only interesting case is that of permutation monomials, for which we set
$(\sigma', \alpha'_{1},\dotsc, \alpha'_{n})(\sigma, \alpha_{1},\dotsc, \alpha_{n})\eqdef
(\sigma'\sigma,\alpha'_{\sigma(1)}\alpha_{1},\dotsc,\alpha'_{\sigma(n)}\alpha_{n})$.
And the identity $id_r$ on $r$ is the same as $r$,
with each variable occurrence $x$ replaced with $id_x$,
and with the identity permutation attached with each monomial.
Inverses are also defined inductively, the key case of monomials being:
$(\sigma, \alpha_{1},\dotsc, \alpha_{n})^{-1}
\eqdef(\sigma^{-1}, \alpha^{-1}_{\sigma^{-1}(1)},\dotsc, \alpha^{-1}_{\sigma^{-1}(n)})$.

If $\vec a = (a_{1},\dotsc, a_{n})$ and $\vec a' = (a'_{1},\dotsc, a'_{n})$,
we set $\vD(\vec{a}, \vec{a}') \eqdef \prod_{i=1}^n \D(a_{i}, a'_{i})$:
with rigid monomials as objects, we obtain a groupoid $\vD$,
which is the free strict monoidal category over $\D$.
Moreover, $\Dm(\vec a,\vec a')=
\sum_{\sigma\in \permutations{n}} \vD(\vec a, \lact{\sigma^{-1}}{\vec a'})$:
$\Dm$ is the free symmetric strict monoidal category over $ \D$.
We call \definitive{quasi-stabilizer} of $\vec{a}$ the subgroup of $\permutations{n}$ defined by
\[
  \qSt(\vec{a}) \eqdef
  \{ \sigma\in \permutations{n} \mid \text{for } 1\le i\le n,\ a_{i} \cong a_{\sigma(i)} \}\quad.
\]
Observe that
$\qSt(\vec{a}) = \St((\softify{a_1},\dotsc,\softify{a_n}))$
and $\sigma\in \qSt(\vec a)$ iff $\vD(\vec a,\lact{\sigma^{-1}}\vec a)\not=\emptyset$.

Let us write $\G(r)$ for the group of automorphisms of $r$: $\G(r)\eqdef\G(r,r)$.
Similarly, we will write $\vD(\vec a)\eqdef\vD(\vec a,\vec a)$.

\begin{lem}\label{lem:cardGmonomial}
For any $ \vec{a}=(a_{1},\dots, a_{n}) \in D^{\oc}$, $\card(\Dm(\vec{a}))
= \card( \qSt(\vec{a})\times\vD(\vec{a}))$.
\end{lem}

\begin{proof}
	Since $\G$ is a groupoid,
	for any morphism $\epsilon:r\cong r'$,
	postcomposition by $\epsilon$ defines
	a bijection from $\G(r)$ to $\G(r,r')$.
	It follows that $\Dm(\vec{a})
	=\sum_{\sigma\in \permutations{n}} \vD(\vec a,\lact{\sigma^{-1}}{\vec a})
	=\sum_{\sigma\in \qSt(\vec{a})}\prod_{i=1}^n \D(a_i, a_{\sigma(i)})$
	is in bijection with
	$\sum_{\sigma\in \qSt(\vec{a})}\prod_{i=1}^n \D(a_i)
	=\qSt(\vec{a})\times\vD(\vec a)$.
\end{proof}

We are then able to formalize the interpretation of the multiplicity of
a resource term $s$ as the number of permutations of monomials
in $s$ leaving any of its writings $a\vartriangleleft s$ unchanged:
\begin{lem}\label{lem:m_is_cardG}
Let $e\in \resExprs$ and let $ r \vartriangleleft e $. Then $m(e) = \card( \G(r))$.
\end{lem}
\begin{proof}
By induction on the structure of $e$. We prove the multiset case.
Assume $e=\bar s$ and $\vec a=(a_1,\dotsc,a_n) \vartriangleleft \bar s$.
Then we can write $\bar s=[s_1,\dotsc,s_n]$ so that $a_i\lhd s_i$
and the induction hypothesis gives $m(s_i)=\card(\G(a_i))$ for $1\le i\le n$.
Then $m(e)=\card(\St((s_1,\dotsc,s_n)))\prod_{i=1}^n\card(\G(a_i))=
\card(\qSt(\vec a))\times \card(\vD(\vec a))$,
and we conclude by Lemma~\ref{lem:cardGmonomial}.
\end{proof}

\subsection{Rigid substitution}
For any $r\in\rigidExprs$ and $\vec b\in\rigidMonomials$ such that $\length{\vec b}=n_x(r)=n$,
we define the \definitive{$n$-linear substitution} $\subst rx{\vec{b}}$ of $\vec b$ for $x$ in $r$
inductively as follows:
\begin{align*}
	\subst xx{(b)}&\eqdef b
	\\
	\subst yx{()} &\eqdef y
	\\
	\subst{(\linj{a})}x{\vec b}
	&\eqdef
	\linj{\subst a x{\vec b}}
	\\
	\subst{(\rinj{a})}x{\vec b}
	&\eqdef
	\rinj{\subst a x{\vec b}}
	\\
	\subst{(\lambda z. a)}x{\vec{b}}
	&\eqdef
	\lambda z. \subst{a}x{\vec{b}}
	\\
	\subst{\rappl c{\vec d\,}}x{\vec b_0\cons\vec b_1}
	&\eqdef
	\rappl{\subst cx{\vec b_0}}{\subst{\vec d\,}x{\vec b_1}}
	\\
	\subst{(a_{1},\dotsc, a_{n})}x{\vec{b}_{1}\cons \cdots\cons \vec{b}_{n}}
	&\eqdef
	( \subst{a_{1}}x{\vec{b}_{1}},\dotsc, \subst{a_{n}}x{\vec{b}_{n}}\})
\end{align*}
where we assume that $y\not=x$, $z\notin \{x\}\cup FV(\vec b)$,
$\length{\vec b}=n_x(a)$,
$\length{\vec b_0}=n_x(c)$,
$\length{\vec b_1}=n_x(\vec d)$,
and $\length{\vec b_i}=n_x(a_i)$ for $1\le i\le n$.

Observe that this substitution is only partially defined.
In order to deal with the general case, we will use the
nullary sum of rigid expressions $0\in\sums{\rigidExprs}$:
again, we consider all the syntactic constructs to be linear
so that we may write, e.g., $\lambda x.a$ for $a\in\rigidTerms\cup\{0\}$
with $\lambda x.0=0$.
We call \definitive{partial rigid expressions} the elements of $\rigidExprs\cup \{0\}$:
we generally use the same typographic conventions for
partial expressions as for regular ones.

Whenever $r\in \rigidExprs\cup\{0\}$
and $\vec{b}\in\rigidMonomials\cup\{0\}$,
we define the \definitive{rigid substitution} $\subst{r}x{\vec{b}}$ of $\vec{b}$
for the variable $x$ in $r$
as above if
$r\in \rigidExprs$, $\vec{b}\in\rigidMonomials$ and $n_{x}(r)=\length{\vec b}$,
and set $\subst rx{\vec b}\eqdef 0$ otherwise.

This rigid version of multilinear substitution will allow us to provide a more
formal account of the intuitive definition of the symmetric multilinear
substitution  $\nsubst ex{\bar u}$, given in Section~\ref{sect:resourceTerms}:
having fixed rigid representations $r\lhd e$ and $\vec b\lhd\bar t=[t_1,\dotsc,t_n]$ with $n=n_x(e)$,
instead of the ambiguous
\[
			\sum\limits_{\sigma\in \permutations{n}} e[ t_{\sigma(1)}/x_{1},\dotsc, t_{\sigma(n)}/x_{n} ]
\]
we can write
\[
\sum_{\sigma\in\permutations{n}}\softify[\big]{\subst rx{\lact \sigma{\vec b}}}\quad .
\]
To prove that this coincides with the inductive definition of $\nsubst ex{\bar u}$,
we need to study how the elements of $\vec b$ are routed to subexpressions of $r$
in the substitution $\subst rx{\lact \sigma{\vec b}}$.

For this, we will rely on the following constructions on permutations.
First, if $ \sigma \in \permutations{n}$ and $\tau \in \permutations{p}$,
we define the \definitive{concatenation} $ \sigma \otimes \tau \in \permutations{n+p}$ by:
\[
	(\sigma\otimes\tau)(i)\eqdef \sigma(i)
	\qquad
  \text{and}
  \qquad
	(\sigma\otimes\tau)(n+j)\eqdef n+\tau(j)
\]
for $1\le i\le n$ and $1\le j\le p$.
This operation is associative and, more generally, we obtain
$\tau_{1}\otimes\cdots\otimes\tau_{n}\in\permutations{k_1+\cdots+k_n}$
whenever
$ \tau_{1} \in \permutations{k_{1}},\dotsc, \tau_{n} \in \permutations{k_{n}}$.
The tensor product notation is justified since, in the category
$\catP$ of natural numbers and permutations, the concatenation of
permutations defines a tensor product
(which is the sum of natural numbers on objects).

Moreover, for each $(I_1,\dotsc,I_n)\in\wcomp{n}{k}$,
writing $I_j=\{i^j_1,\dotsc,i^j_{k_j}\}$ with $i^j_1<\cdots<i^j_{k_j}$,
we set $\gamma_{I_1,\dotsc,I_n}(i^j_l)\eqdef l+\sum_{r=1}^{j-1} k_r$:
then $\gamma_{I_1,\dotsc,I_n}$
is the unique permutation $\gamma\in\permutations{k}$ such that
the map $(j,l)\mapsto \gamma(i^j_l)$ is strictly increasing,
considering the lexicographic order on pairs.

Given a weak $n$-composition of $k$, \emph{i.e.} a tuple
$(k_1,\dotsc,k_n)\in\N^n$ such that $k=\sum_{j=1}^n k_j$,
we write $\wcompk{n}{k_1,\dotsc,k_n}{k}$
for the set of those $(I_1,\dotsc,I_n)\in\wcomp{n}{k}$
such that $\card(I_i)=k_i$ for $1\le i\le n$.
We obtain:
\begin{lem}%
  \label{lem:permComp}
  For any weak $n$-composition $(k_1,\dotsc,k_n)$ of $k$,
  the function
  \begin{align*}
    \wcompk{n}{k_1,\dotsc,k_n}{k}\times \prod_{j=1}^n\permutations{k_j}
    &\to \permutations{k}
    \\
    ((I_1,\dotsc,I_n),(\sigma_1,\dotsc,\sigma_n))
    &\mapsto(\sigma_1\otimes\cdots\otimes\sigma_n)\gamma_{I_1,\dotsc,I_n}
  \end{align*}
  is bijective.
\end{lem}
\begin{proof}
  The inverse function is as follows: given $\sigma\in\permutations{k}$,
  we fix $I_j\eqdef\{i\in\{1,\dotsc,k\}\mid \sum_{r=1}^{j-1} k_r< \sigma(i)\le \sum_{r=1}^{j} k_r\}$;
  then, using the above notations for the elements of $I_j$,
  for each $l\in\{1,\dotsc,k_j\}$, we fix $\sigma_j(l)\in\{1,\dotsc,k_j\}$
  to be the unique $l'$ such that $\sigma(i^j_l)=l'+\sum_{r=1}^{j-1} k_r$.
\end{proof}

Now we can show that the two definitions of symmetric multilinear substitution coincide:
\begin{lem}%
	\label{lem:softifySubst}
	If $r\lhd e$ and $\vec b\lhd \bar t$ then $n_x(r)=n_x(e)$
	and $\length{\vec b}=\length{\bar t}$.
	Moreover
	$\nsubst ex{\bar t}=\sum_{\sigma\in\permutations{\length{\vec b}}}\softify[\big]{\subst rx{\lact \sigma{\vec b}}}$.
\end{lem}

\begin{proof}
	The first two identities follow directly from the definitions.
	If $n_x(r)\not=\length{\vec b}$ then both sides of the third identity are $0$.
	Otherwise, it is proved by induction on $r$.

	Let us treat the case of a monomial:
	write $r=(a_1,\dotsc,a_n)$ and $e=[s_1,\dotsc,s_n]$
	with $a_i\lhd s_i$ for $1\le i\le n$.
	Then
	\begin{align*}
		\nsubst ex{\bar t}
    &= \sum_{(I_1,\dotsc,I_n)\in\wcomp{n}{\length{\vec b}} }
		[\nsubst {s_1}{x}{\bar t_{I_1}},\dotsc,\nsubst {s_n}{x}{\bar t_{I_n}}]
		\\
    &= \sum_{(I_1,\dotsc,I_n) \in\wcompk{n}{k_1,\dotsc,k_n}{\length{\vec b}}}
		[\nsubst {s_1}{x}{\bar t_{I_1}},\dotsc,\nsubst {s_n}{x}{\bar t_{I_n}}]
	\end{align*}
	where we write $k_i=n_x(s_i)$ for $1\le i\le n$.

	If $I\subseteq \{1,\dotsc,\length{\vec b}\}$
	then we write $\vec b_I=(b_{i_1},\dotsc,b_{i_k})$
	where $i_1<\cdots<i_k$ enumerate $I$.
	By induction hypothesis we obtain
	\begin{align*}
		\nsubst ex{\bar t}
		&=
    \sum_{(I_1,\dotsc,I_n) \in\wcompk{n}{k_1,\dotsc,k_n}{\length{\vec b}}}
		\bigg[
			\sum_{\sigma_1\in\permutations{k_1}}
			\softify[\big]{\subst {a_1}x{\lact{\sigma_1}{\vec b_{I_1}}}}\,
			,\dots,
			\sum_{\sigma_n\in\permutations{k_n}}
			\softify[\big]{\subst {a_n}x{\lact{\sigma_n}{\vec b_{I_n}}}}
		\bigg ]
		\\
		&=
    \sum_{(I_1,\dotsc,I_n) \in\wcompk{n}{k_1,\dotsc,k_n}{\length{\vec b}}}
		\sum_{\sigma_1\in\permutations{k_1}}
		\cdots
		\sum_{\sigma_n\in\permutations{k_n}}
		\softify[\big]
		{
			\subst[\big] {r}x{\lact{\sigma_1}{\vec b_{I_1}}\cons\cdots\cons \lact{\sigma_n}{\vec b_{I_n}}}
		}
	\end{align*}
	and we conclude, observing that
  $\lact {\sigma_1}{\vec b_{I_1}}\cons\cdots\cons \lact {\sigma_n}{\vec b_{I_n}}
  =\lact{(\sigma_1\otimes\cdots\otimes\sigma_n)\gamma_{I_1,\dotsc,I_n} }{\vec b}$,
  hence the families
	\[
		\big( \lact {\sigma_1}{\vec b_{I_1}}\cons\cdots\cons \lact {\sigma_n}{\vec b_{I_n}} \big)_{\substack{
      (I_1,\dotsc,I_n)\in \wcompk{n}{k_1,\dotsc,k_n}{\length{\vec b}}
			,\
			(\sigma_1,\dotsc,\sigma_n)\in\permutations{k_1}\times \cdots\times \permutations{k_n}
		}}
	\]
	and
	$\big( \lact{\sigma}{\vec b} \big)_{\sigma\in\permutations{\length{\vec b}}}$
  coincide up to reindexing \emph{via} the bijection of Lemma~\ref{lem:permComp}.
\end{proof}
Informally, everything thus works out as if
$[s_1,\dotsc,s_n]=\sum_{\sigma\in\permutations{n}}(s_1,\dotsc,s_n)$,
which is to be related with the $\frac{1}{n!}$ coefficient in the Taylor expansion,
cancelling out the cardinality of $\permutations{n}$.
Forgetting about coefficients, we obtain:
\begin{cor}\label{cor:lhdSubst}
	If $ r \lhd e $ and $ \vec{b} \lhd \bar{t}$ with $n_x(e)=\length{\bar t}$,
	then  $ \supp(\partial_{x} e \cdot \bar{t})
	= \{ \softify[\big]{\subst rx{\lact \sigma{\vec b}}} \mid \sigma\in\permutations{\length{\vec b}}\}$.
\end{cor}

Conversely, any rigid representative of a symmetric substitution is
obtained as a rigid substitution:
\begin{lem}\label{lem:lhdDeSubst}
	If $r'\lhd e'\in \supp(\partial_{x} e \cdot \bar{t})$
	then $n_x(e)=\length{\bar t}$ and
	there exist $ r \lhd e $ and $ \vec{b} \lhd \bar{t}$
	such that $r'=\subst rx{\vec b}$.
\end{lem}
\begin{proof}
	By induction on $e$.
If $e = x $ then $ \bar{t} = [t] $ for some $ t \in\resTerms$ and $e' = t $.
If $r'\lhd e'=t$ then we can set $r=x$ and $\vec b=(r')$.
If $e=y\not=x$ then $\bar t=[]$ and we can set $r=y$ and $\vec b=()$.
The abstraction and sum cases follow immediately from the induction hypothesis.

If $ e = \rappl {s} {\bar{v}} $, we write $\bar t=[t_1,\dotsc,t_n]$ and obtain
\[  \nsubst {e} {x}  {\bar{t}} = \sum_{(I_{1}, I_{2})\in\wcomp{2}{n}} \rappl { \nsubst  {s}{x} {\bar{t}_{I_{1}}} } {\nsubst  {\bar{v}} {x} {\bar{t}_{I_{2}} }} \quad.\]
Then $e'=\rappl {s' }{\bar v'}$ with $ s' \in \supp(\nsubst  {s}{x} {\bar{t}_{I_{1}}}) $
and $ \bar v' \in \supp(\nsubst  {\bar{v}} {x} {\bar{t}}_{I_{2}} ) $
for some $ (I_{1}, I_{2})\in\wcomp{2}{n}$.
It follows that $r'=\rappl{a}{\vec d}$ with $a\lhd s'$ and $\vec d\lhd \bar v'$.
By induction hypothesis, we obtain $c_1\lhd s$, $\vec{b}_{1} \lhd \bar{t}_{I_{1}} $,
$\vec c_2 \lhd \bar v$ and $ \vec{b}_{2} \lhd \bar{t}_{I_{2}} $
such that $a = \subst {c_1} {x} {\vec{b}_{1}} $ and $ \vec d = \subst {\vec c_2} {x} {\vec{b}_{2}} $.
Then we conclude by setting $r=\rappl{c_1}{\vec c_2}\lhd \rappl s{\bar v}=e$
and $\vec b = \vec{b}_{1} :: \vec{b}_{2} \lhd \bar{t}_{I_{1}} \cdot \bar{t}_{I_{2}}=\bar t $.

The case of monomials is similar.
\end{proof}

\subsection{Substitution for permutation expressions}

The key intermediate result for Step~\ref{step:coefNF} is the fact that if
$e\coh e$ and
$e'\in \supp(\nsubst ex{\bar t})$ then $(\nsubst ex{\bar t})_{e'}=\frac{m(e)m(\bar t)}{m(e')}$:
this will be established in Lemma~\ref{lem:coeffNsubst},
which concludes the present section.
With that goal in mind, and having characterized $m(e)$ as the cardinality of
the group $\G(r)$ for any $r\lhd e$, it becomes essential to study how the
automorphisms of $r'\lhd e'\in \supp(\nsubst ex{\bar t})$ are related
with those of some $r\lhd e$ and $\vec b\lhd \bar t$:
by Lemma~\ref{lem:lhdDeSubst}, we can choose $r$ and $\vec b$ such that
$r'=\subst{r}x{\vec b}$.
Then it seems natural to consider some form of substitution for
permutation expressions, following the structure of rigid substitution.

We define the \emph{substitution of permutation terms for a variable} as follows.
Given $\epsilon \in \G(r,r')$
and $ \vec{\beta} \in \vD(\vec{b},  \vec{b}')$
with $\length{\vec b} = n_{x}(r)$,
we construct $\epsilon[\vec{\beta} / x]$
by induction on $\epsilon$:
\begin{align*}
	(id_{x})[(\beta)/x]
	&\eqdef \beta
	\\
	(id_{y})[()/x]
	&\eqdef id_{y}
	\\
	(\lambda y.\alpha)[\vec{\beta}/x]
	&\eqdef \lambda y. \alpha[\vec{\beta}/x]
	\\
 	(\linj \alpha)[\vec{\beta}/x]
	&\eqdef \linj{\alpha[\vec{\beta}/x]}
	\\
 	(\rinj \alpha)[\vec{\beta}/x]
	&\eqdef \rinj{\alpha[\vec{\beta}/x]}
	\\
	(\rappl{\gamma}{\tilde\delta})[\vec\beta_1\cons\vec\beta_2/x]
	&\eqdef \rappl{\gamma[\vec\beta_1/x]}{\tilde\delta[\vec\beta_2/x]}
	\\
	(\sigma,(\alpha_1,\dotsc,\alpha_n))[\vec \beta_1\cons\cdots\cons\vec \beta_n/x]
	&\eqdef (\sigma,(\alpha_1[\vec\beta_1/x],\dotsc,\alpha_n[\vec\beta_n/x]))
\end{align*}
where we assume that $y\not=x$, $z\notin \{x\}\cup FV(\vec \beta)$,
$\length{\vec\beta_1}=n_x(\gamma)$,
$\length{\vec\beta_2}=n_x(\delta)$,
and $\length{\vec\beta_i}=n_x(\alpha_i)$ for $1\le i\le n$.

If $\epsilon \in \G(r,r')$ and $ \vec{\beta} \in \vD(\vec{b},  \vec{b}')$,
the source of $\subst\epsilon x{\vec\beta}$
is obviously $\subst rx{\vec b}$ but describing
its target is more intricate:
in general, $\epsilon[\vec\beta/x]\not\in\G(r[\vec b/x],r'[\vec b'/x])$.
\begin{exa}
Consider the rigid monomials $\vec a= (x,x)$ and $ \vec{b}=(\langle z \rangle (), \langle z \rangle (z))$.
Writing $\tau$ for the unique transposition of $\permutations{2}$,
we obtain $ \alpha= (\tau, id_{x}, id_{x}) \in \Dm(\vec a)$.
Let $\vec{\beta}
= (id_{\rappl{z}{()}}, id_{\rappl{z}{(z)}}) \in \vD(\vec{b})$.
Then $\alpha[\vec{\beta}/x] = (\tau, id_{\rappl{z}{()}}, id_{\rappl{z}{(z)}})$,
hence $\alpha[\vec{\beta}/x]:a[\vec b/x]\cong( \langle z \rangle (z), \langle z \rangle ())\not=a[\vec b/x]$.
\end{exa}

To describe the image of $r[\vec b/x]$ through $\epsilon[\vec\beta/x]$,
we first introduce another operation on permutations.
If $ \sigma \in \permutations{n}$ and
$ \tau_{i} \in \permutations{k_{i}}$
for $1\le i\le n$,
we define the \definitive{multiplexing}
$ \sigma \cdot ( \tau_{1},\dots, \tau_{n}) \in \permutations{k_{1} + \cdots + k_{n}}$
by:
\[
	(\sigma \cdot ( \tau_{1},\dots,  \tau_{n}))\bigg (l + \sum_{j=1}^{i-1} k_j\bigg)
	\eqdef \tau_{i}(l) + \sum_{j=1}^{\sigma(i)-1} k_{\sigma^{-1}(j)}
\]
for $1\le i\le n$ and $1\le l\le k_i$.
Multiplexing may be described in the category $\catP$ of natural numbers and permutations,
which is symmetric strict monoidal, as follows:
$ \sigma \cdot ( \tau_{1},\dots, \tau_{n})
= \sigma_{k_1,\dotsc,k_n} \circ (\tau_{1}\otimes\cdots\otimes\tau_{n}) $
where $\sigma_{k_1,\dotsc,k_n}$ is the canonical symmetry map
$k_1+\cdots+k_n\to k_{\sigma^{-1}(1)}+\cdots+k_{\sigma^{-1}(n)}=k_1+\cdots+k_n$
associated with the left action of $\sigma$ on $n$-ary tensor products in $\catP$.
This decomposition of multiplexing is depicted in Figure~\ref{fig:multiplexing}.

\begin{figure}[t]
	\begin{center}
	\begin{tikzpicture}[
			permutation/.style={draw=black,rectangle,minimum width=2cm,minimum height=1 cm,outer xsep=1em},
			port/.style={circle,fill=black,minimum size=1 mm},
	]
	\node[permutation](tau1) at (-3,0) {$\tau_1$};
	\node[above=.5cm of tau1,inner xsep=-.5ex] (in1) {$\overbrace{1+\cdots+1}^{k_1}$};
	\node[left=.5em of in1.south west,anchor=south east]{(};
	\node[right=.5em of in1.south east,anchor=south west] (plusl) {$)\quad+$};
	\draw(in1.south west)--(in1.south west |- tau1.north) node[midway] (in1l) {};
	\draw(in1.south east)--(in1.south east |- tau1.north) node[midway] (in1r) {};
	\draw[dotted](in1l)--(in1r);

	\node[permutation](taun) at (3,0) {$\tau_n$};
	\node[above=.5cm of taun,inner xsep=-.5ex] (inn) {$\overbrace{1+\cdots+1}^{k_n}$};
	\node[left=.5em of inn.south west,anchor=south east] (plusr) {$+\quad($};
	\node[right=.5em of inn.south east,anchor=south west]{)};
	\draw(inn.south west)--(inn.south west |- taun.north) node[midway] (innl) {};
	\draw(inn.south east)--(inn.south east |- taun.north) node[midway] (innr) {};
	\draw[dotted](innl)--(innr);

	\draw[dotted,very thick](tau1)--(taun);
	\path (plusl)--(plusr) node[midway] {$\cdots$};

	\node[permutation,very thick,minimum width=8cm] (sigma) at (0,-1.5) {$\sigma_{k_1,\dotsc,k_n}$};

	\coordinate (out1l1) at (in1.west |- tau1.south);
	\draw (out1l1)--(out1l1 |- sigma.north) node [midway] (out1l) {};
	\coordinate (out1r1) at (in1.east |- tau1.south);
	\draw (out1r1)--(out1r1 |- sigma.north) node [midway] (out1r) {};
	\draw[dotted](out1l)--(out1r);

	\coordinate (outnl1) at (inn.west |- taun.south);
	\draw (outnl1)--(outnl1 |- sigma.north) node [midway] (outnl) {};
	\coordinate (outnr1) at (inn.east |- taun.south);
	\draw (outnr1)--(outnr1 |- sigma.north) node [midway] (outnr) {};
	\draw[dotted](outnl)--(outnr);

	\coordinate (sig1) at (tau1 |- sigma.south);
	\node[below=.5cm of sig1,inner xsep=-.5ex] (dest1) {$\underbrace{1+\cdots+1}_{k_{\sigma^{-1}(1)}}$};
	\node[left=.5em of dest1.north west,anchor=north east]{(};
	\node[right=.5em of dest1.north east,anchor=north west] (plusl) {$)\quad+$};

	\coordinate (sign) at (taun |- sigma.south);
	\node[below=.5cm of sign,inner xsep=-.5ex] (destn) {$\underbrace{1+\cdots+1}_{k_{\sigma^{-1}(n)}}$};
	\node[left=.5em of destn.north west,anchor=north east] (plusr) {$+\quad($};
	\node[right=.5em of destn.north east,anchor=north west] {)};
	\path (plusl)--(plusr) node[midway] {$\cdots$};

	\coordinate (sig1l1) at (out1l1 |- sigma.south);
	\coordinate (sig1r1) at (out1r1 |- sigma.south);
	\coordinate (signl1) at (outnl1 |- sigma.south);
	\coordinate (signr1) at (outnr1 |- sigma.south);
	\draw (sig1l1)--(sig1l1 |- dest1.north) node [midway] (sig1l) {};
	\draw (sig1r1)--(sig1r1 |- dest1.north) node [midway] (sig1r) {};
	\draw (signl1)--(signl1 |- destn.north) node [midway] (signl) {};
	\draw (signr1)--(signr1 |- destn.north) node [midway] (signr) {};
	\draw[dotted](sig1l)--(sig1r);
	\draw[dotted](signl)--(signr);

	\end{tikzpicture}
	\end{center}

\caption{Graphical representation of $\sigma\cdot(\tau_1,\dotsc,\tau_n)$%
\label{fig:multiplexing}}

\end{figure}

Multiplexed permutations compose as follows:
\begin{lem}\label{lem:compMultiplex}
If $ \sigma,\sigma' \in \permutations{n}$,
$ \tau_{i} \in \permutations{k_{i}}$
and
$ \tau'_{i} \in \permutations{k_{\sigma^{-1}(i)}}$
for $1\le i \le n$,
then
\[
	\big(\sigma' \cdot ( \tau'_{1},\dots, \tau'_{n}) \big)
	\big(\sigma \cdot ( \tau_{1},\dots, \tau_{n}) \big)
  =
	(\sigma'\sigma)\cdot(\tau'_{\sigma(1)}\tau_1,\dotsc, \tau'_{\sigma(n)}\tau_n)
\]
and
\[
	\big(\sigma \cdot ( \tau_{1},\dots, \tau_{n}) \big)^{-1}
  =
	\sigma^{-1}\cdot(\tau_{\sigma^{-1}(1)}^{-1},\dotsc, \tau_{\sigma^{-1}(n)}^{-1})
	\quad.
\]
\end{lem}
\begin{proof}
	We detail the proof only in case the result is not obvious to the reader from
	the above categorical presentation of multiplexing.
	Let $\alpha=\sigma \cdot ( \tau_{1},\dots, \tau_{n})$
	and $\alpha'=\sigma' \cdot ( \tau'_{1},\dots, \tau'_{n})$.
	For $1\le i\le n$ and $1\le l\le k_i$:
\begin{align*}
	\alpha'\bigg (\alpha\big (
	\sum_{j = 1}^{i-1} k_j + l
	\big )
	\bigg )
	&=
	\alpha'\bigg(
	\sum_{j = 1}^{\sigma(i)-1} k_{\sigma^{-1}(j)}
	+ \tau_i(l)
	\bigg )
	\\
	&=
	\sum_{j = 1}^{\sigma'(\sigma(i))-1} k'_{{\sigma'}^{-1}(j)}
	+ \tau'_{\sigma(i)}(\tau_i(l))
	&&\text{(writing $k'_i=k_{\sigma^{-1}(i)}$)}
	\\
	&=
	\sum_{j = 1}^{(\sigma'\sigma)(i)-1} k_{(\sigma'\sigma)^{-1}(j)}
	+ (\tau'_{\sigma(i)}\tau_i)(l)
\end{align*}
which establishes the first identity.
The second identity follows directly.
\end{proof}

The action of multiplexed permutations on sequences is as follows:
\begin{lem}\label{lem:actionMultiplex}
Let $\vec{b}, \vec{b}_{1}, \dots, \vec{b}_{n} \in D^{!}$, $ \sigma \in \permutations{n}$ and $ \tau_{i} \in \permutations{\length {\vec{b}_{i}}} $ for all $i \in \{1, \dots, n \}$. If $\vec{b} = \vec{b}_{1} :: \cdots :: \vec{b}_{n}$  then
$ \lact{\sigma \cdot ( \tau_{1},\dots, \tau_{n})}{\vec{b}}
= \lact{\tau_{\sigma^{-1}(1)}}{\vec{b}_{\sigma^{-1}(1)}}
:: \cdots :: \lact{\tau_{\sigma^{-1}(n)}}{\vec{b}_{\sigma^{-1}(n)}}$.
\end{lem}

\begin{proof}
	Again, we detail the proof only in case the result is not obvious from
	the categorical presentation.
	Set $\length {\vec{b}_{i}} = k_{i} $, so that $\length {\vec b} = \sum_{i = 1}^{n} k_{i}$.
	Write $\vec b'= \lact{\sigma \cdot ( \tau_{1},\dots, \tau_{n})}{\vec{b}}$.
	For $1\le p\le \length{\vec b'}=\length{\vec b}=\sum_{j=1}^{n} k_{\sigma^{-1}(j)}$,
	we can write $p =  \sum_{j=1}^{i-1} k_{\sigma^{-1}(j)} + l$
	with $ i \in \{ 1,\dots, n \} $ and $ l \in \{ 1,\dots, k_{\sigma^{-1}(i)} \}$.
	Then, by Lemma~\ref{lem:compMultiplex}, $(\sigma \cdot ( \tau_{1},\dots, \tau_{n}))^{-1}(p)= \sum_{j=1}^{\sigma^{-1}(i)-1} k_j+\tau_{\sigma^{-1}(i)}^{-1}(l)$
	and $b'_{p}=b_{(\sigma \cdot ( \tau_{1},\dots, \tau_{n}))^{-1}(p)}=(\vec b_{\sigma^{-1}(i)})_{\tau_{\sigma^{-1}(i)}^{-1}(l)}=(\lact{\tau_{\sigma^{-1}(i)}}{\vec b_{\sigma^{-1}(i)}})_l$.
\end{proof}

We can now define the \definitive{restriction} $\restr{\epsilon}{x} \in \permutations{n_{x}(r)}$
of $\epsilon\in \G(r,r')$ to the occurrences of $x$ in $r$, by induction on $\epsilon$:
\begin{align*}
  \restr{id_{x}}{x}
	&\eqdef id_{\{1\}}
	\\
  \restr{id_{y}}{x}
	&\eqdef id_\emptyset
	\\
	\left.
	\begin{array}{r}
    \restr{(\lambda y. \alpha)}{x}\\
    \restr{(\linj \alpha)}{x}\\
    \restr{(\rinj \alpha)}{x}
	\end{array}
	\right\}
  &\eqdef \restr{\alpha}{x}
	\\
  \restr{(\rappl{\gamma}{\tilde\delta})}{x}
  &\eqdef \restr{\gamma}{x} \otimes \restr{\tilde\delta}{x}
	\\
  \restr{(\sigma, \alpha_1,\dotsc,\alpha_n)}{x}
  &\eqdef \sigma\cdot (\restr{\alpha_1}{x}, \dots, \restr{\alpha_n}{x})
\end{align*}
where we assume $x\not=y$.
Intuitively $\restr{\epsilon}{x}$ is the permutation
induced by $\epsilon$ on the occurrences $x_1,\dotsc,x_{n_x(r)}$ of $x$ in $r$,
taken from left to right.

We recall that $\catP$ denotes the category of finite cardinals and permutations.
For any variable $x$,
we define an application $F_{x}$ from $\G $ to $\catP$ as follows:
$F_{x}(r)\eqdef n_{x}(r)$ and $F_{x}(\alpha)\eqdef \restr{\alpha}{x}$.

\begin{lem}%
	\label{lem:functor}
$F_{x}$ is a functor from $\G$ to $\catP$.
\end{lem}
\begin{proof}
By induction on permutation expressions. We focus on the composition condition for the list case.
Let $\tilde\alpha : \vec{a} = (a_{1}, \dots, a_{n}) \cong \vec{b} = ( b_{1}, \dots, b_{n})$
and $\tilde\beta : \vec{b} \cong \vec{c} = (c_{1},\dots, c_{n})$.
By definition $\tilde\alpha = ( \sigma, \alpha_{1}, \dots, \alpha_{n})$ and
$\tilde\beta = ( \tau, \beta_{1}, \dots, \beta_{n})$,
for some $\sigma, \tau $ in $ \permutations{n}$ and with
$ \alpha_{i} : a_{i} \cong b_{\sigma(i)} $ and $ \beta_{i} : b_{i} \cong c_{\tau(i)}$.
The composition $\tilde\beta\tilde\alpha $ is then defined as the isomorphism
$ (\tau \sigma,  \beta_{\sigma(1)} \alpha_{1}, \dots \beta_{\sigma(n)} \alpha_{n})$.

We have to prove that
$\restr{(\tilde\beta\tilde\alpha)}{x}=\restr{\tilde\beta}{x}\restr{\tilde\alpha}{x}$,
that is
\[
	(\tau\sigma)\cdot\big (
    \restr{(\beta_{\sigma(1)} \alpha_1)}{x}
		,\dotsc,
    \restr{(\beta_{\sigma(n)} \alpha_n)}{x}
	\big)
  =(\tau\cdot(\restr{\beta_1}{x}, \dots, \restr{\beta_n}{x}))
  (\sigma\cdot(\restr{\alpha_1}{x}, \dots, \restr{\alpha_n}{x}))
\]
which is a direct consequence of the inductive hypothesis,
$\restr{(\beta_{\sigma(i)}\alpha_i)}{x} = \restr{\beta_{\sigma(i)}}{x}\restr{\alpha_i}{x} $
for $1\le i\le n$,
\emph{via} Lemma~\ref{lem:compMultiplex}.
\end{proof}
In particular, the restriction of $F_{x}$ to the automorphism group of some
rigid expression $r$ is a group homomorphism from $\G(r)$ to $\permutations{n_x(r)}$:
its image $\restr{\G(r)}{x}$ is thus a subgroup of $\permutations{n_x(r)}$.
This homomorphism will play a crucial rôle in Section~\ref{sect:rigid:cohSubst}.

This operator allows us to
describe the image of $\epsilon[\vec\beta/x]$ as follows:
\begin{lem}\label{lem:substPerm}
	If $\epsilon:r\cong r'$ and $\vec\beta\in\vD(\vec b,\vec b')$
	with $\length{\vec\beta}=n_x(r)$ then
	$\epsilon[\vec\beta/x]:r[\vec{b}/x]\cong r'[\lact{\restr{\epsilon}{x}}{\vec{b}'}/x]$.
\end{lem}

\begin{proof}
	By induction on the structure of $r$. The interesting case is the list case.
	Assume $r=(a_{1},\dotsc, a_{n})$,
	$r'=(a'_{1},\dotsc, a'_{n})$,
	$\epsilon = (\sigma, \alpha_{1},\dotsc, \alpha_{n})$
	and $\vec \beta=\vec \beta_1\cons\cdots \cons\vec\beta_n$,
	with
	$\alpha_i:a_i\cong a'_{\sigma(i)}$,
	$\vec b=\vec b_1\cons \cdots \cons \vec b_n$,
	$\vec b'=\vec b'_1\cons \cdots \cons \vec b'_n$,
	$\length{\vec \beta_i}=n_x(a_i)$
	and $\vec \beta_i\in \vD(\vec{b}_i,  \vec{b}'_i)$.
	By definition,
	we have
	$\alpha[\vec\beta/x] = (\sigma,\alpha_1[\vec\beta_1/x],\dotsc,\alpha_n[\vec\beta_n/x])$.
	Since $\alpha_i:a_i\cong a'_{\sigma(i)}$, we obtain
	$\alpha_i[\vec\beta_i/x] : a_{i}[\vec{b}_{i} / x]\cong  a'_{\sigma (i)}[ \lact{\restr{\alpha_i}{x}}{\vec{b}_{ i}'} / x]$
	by induction hypothesis.

	We obtain
	\begin{align*}
		\subst \alpha x{\vec \beta}:\subst r x {\vec b}&\cong
		\big(
      \subst[\big]{a'_{1}}x{\lact{\restr{\alpha_{\sigma^{-1}(1)}}{x}}{\vec{b}_{\sigma^{-1}(1)}'}}
			,\dotsc,
      \subst[\big]{a'_{n}}x{\lact{\restr{\alpha_{\sigma^{-1}(n)}}{x}}{\vec{b}_{\sigma^{-1}(n)}'}}
		\big)
		\\&=
		\subst[\big]{ r'}x{
      \lact{\restr{\alpha_{\sigma^{-1}(1)}}{x}}{\vec{b}_{\sigma^{-1}(1)}'}
			\cons \cdots \cons
      \lact{\restr{\alpha_{\sigma^{-1}(n)}}{x}}{\vec{b}_{\sigma^{-1}(n)}'}
		}
	\end{align*}
	and we conclude by Lemma~\ref{lem:actionMultiplex}.
\end{proof}

\subsection{The combinatorics of permutation expressions under coherent substitution}%
\label{sect:rigid:cohSubst}

Substitution is injective on parallel permutation expressions,
in the following sense:
\begin{lem}\label{lem:substInj}
Let $r, r' \in D^{(!)}$ and $\vec{b}, \vec{b}' \in D^{!}$
with $\length{\vec b}=n_x(r)$ and $\length{\vec b'}=n_x(r')$,
and let $ \epsilon,\epsilon'\in \G(r,r')$ and $\vec\beta,\vec\beta'\in \vD(  \vec{b}, \vec{b}')$.
If $ \epsilon [\vec{\beta} / x] = \epsilon' [\vec{\beta'} / x]$ then $ \epsilon = \epsilon' $ and $ \vec{\beta} = \vec{\beta }'$.
\end{lem}
\begin{proof}
By a straightforward induction on the structure of $r$.
\end{proof}

On the other hand, surjectivity does not hold in general,
because the substitution might enable new morphisms
$\subst rx{\vec b}\cong\subst{r'}x{\vec b'}$,
not induced by morphisms in $\G(r,r')$ and $\vD(\vec b,\vec b')$:
\begin{exa}
Let $a = \langle \langle y \rangle (x) \rangle \langle z \rangle (x)$, $a'= \langle \langle x \rangle (y) \rangle \langle z \rangle (x)$ and $ \vec{b} = (y,z) $.
Then $ a [\vec{b}/x] = a' [\vec{b}/x]$ but $a\not\cong a'$.
\end{exa}
Observe that, in the above example, $\softify a\not\coh \softify{a'}$.
Indeed, in the following, we will establish that coherence
allows to restore a precise correspondence between the permutation expressions
on a substitution $\subst rx{(b_1,\dotsc, b_n)}$
and the $(1+n)$-tuples of permutation expressions
on $r$ and each of the $b_i$'s respectively.
It will be useful to consider the coherence relation defined on rigid
expressions by the rules of Figure~\ref{fig:rigidCoh},
so that $r\coh r'$ iff $\softify r\coh \softify{r'}$.
\begin{figure}
	\begin{gather*}
		\begin{prooftree}
			\infer0[]{x \coh x}
		\end{prooftree}
		\qquad
		\begin{prooftree}
			\hypo{ a \coh a'}
			\infer1[]{ \lambda x. a \coh \lambda x. a' }
		\end{prooftree}
		\qquad
		\begin{prooftree}
			\hypo{ c \coh c'}
			\hypo{ \vec d \coh \vec d'}
			\infer2[]{ \langle c \rangle \vec d \coh \langle c' \rangle \vec d'}
		\end{prooftree}
		\qquad
		\begin{prooftree}
			\hypo{b_{i} \coh b_{j} \text{ for }1\le i,j\le n + m}
			\infer1[]{ (b_{1},\dotsc, b_{n}) \coh (b_{n + 1},\dotsc, b_{n + m}) }
		\end{prooftree}
		\\[1ex]
		\begin{prooftree}
			\hypo{ a \coh a'}
			\infer1[]{ a \oplus \bullet \coh a' \oplus \bullet }
		\end{prooftree}
		\qquad
		\begin{prooftree}
			\hypo{ a \coh a'}
			\infer1[]{ \bullet \oplus s \coh \bullet \oplus s' }
		\end{prooftree}
		\qquad
		\begin{prooftree}
			\infer0[]{a \oplus \bullet \coh \bullet \oplus a'}
		\end{prooftree}\quad.
	\end{gather*}
	\caption{Rules for the coherence relation on $\rigidExprs$.}%
	\label{fig:rigidCoh}
\end{figure}
Then we obtain:
\begin{lem}\label{lem:substSurj}
Let $r, r' \in D^{(!)}$ and $\vec{b}, \vec{b}' \in D^{!}$
with $\length{\vec b}=n_x(r)$ and $\length{\vec b'}=n_x(r')$.
If $r \coh r' $ then for all $\phi\in \G( r[\vec{b}/x] ,  r' [\vec{b'}/x]) $
there exist  $ \epsilon\in \G(r,r')$ and
$  \vec{\beta} \in\vD(  \vec{b}, \lact{\restr{\epsilon}{x}^{-1}}{\vec{b}'})$ such that $\phi= \epsilon [\vec{\beta} / x]$.
\end{lem}
\begin{proof}
	By induction on the structure of $r$:
	the coherence hypothesis $r\coh r'$ induces that $r$ and $r'$
	are of the same syntactic nature.

	If $r=x$ then $r'=x$ and we can write
	$\vec b=(b)$, $\vec b'=(b')$
	with $\phi:b\cong b'$. Then we set $\epsilon=id_x$ and $\vec\beta=(\phi)$.
	If $r=y\not=x$ then $r'=y$ and $\phi=id_y$,
	and we set $\epsilon=id_y$ and $\vec\beta=()$.
	The abstraction, application and sum cases
	follow straightforwardly from the induction hypotheses.
	We detail the list case.

	We have $ r = ( a_{1}, \dots, a_{n})$ and $ r' = (a'_{1}, \dots, a'_{m})$.
	Since $\phi:r[\vec{b}/x]\cong r' [\vec{b}'/x]$
	we must have $m =n$,
	$\vec b=\vec b_1\cons \cdots\cons \vec b_n$,
	$\vec b'=\vec b'_1\cons \cdots\cons \vec b'_n$
	and $ \phi= ( \sigma, \gamma_{1}, \dots, \gamma_{n})$
	with $\gamma_i \in \G( a_i[\vec b_i/x] ,  a'_{\sigma(i)} [\vec b'_{\sigma(i)}/x]) $.
	Since $ r \coh r' $ we have in particular $a_{i} \coh a'_{\sigma(i)}$ for $1\le i \le n  $.

	By the induction hypothesis, we obtain $\gamma_{i} = \alpha_{i} [\vec{\beta}_{i} / x]$
	with $ \alpha_{i} \in \G(a_{i}, a'_{\sigma(i)})$
  and $\vec{\beta}_{i} \in \vD(\vec{b}_{i}, \lact{\restr{\alpha_i}{x}^{-1}}{\vec{b}'_{\sigma(i)}})$.
	Then by definition $\epsilon\eqdef (\sigma, \alpha_{1}, \dots, \alpha_{n}) : r \cong r'$ and
	\begin{align*}
		\vec \beta \eqdef \vec \beta_{1} :: \cdots :: \vec \beta_{n} : \vec b
		&\cong
    \lact{\restr{\alpha_1}{x}^{-1}}{\vec{b}'_{\sigma(1)}}
    \cons \cdots\cons \lact{\restr{\alpha_n}{x}^{-1}}{\vec{b}'_{\sigma(n)}}
		\\
		&=
    \lact{\sigma^{-1}\cdot (\restr{\alpha_{\sigma^{-1}(1)}}{x}^{-1},\dotsc,\restr{\alpha_{\sigma^{-1}(n)}}{x}^{-1})}{\vec{b}'}
		&&\text{(by Lemma~\ref{lem:actionMultiplex})}
	\end{align*}
	and it remains only to prove that
  $\sigma^{-1}\cdot (\restr{\alpha_{\sigma^{-1}(1)}}{x}^{-1},\dotsc,\restr{\alpha_{\sigma^{-1}(n)}}{x}^{-1})
  =\restr{\epsilon}{x}^{-1}$,
	which follows from Lemma~\ref{lem:compMultiplex}.
\end{proof}

In particular, we obtain
$(\restr{\epsilon}{x},\vec\beta)\in\Dm(\vec b,\vec b')$,
hence:
\begin{cor}%
	\label{cor:cohAntired}
	If $r \coh r' $ and $r[\vec{b}/x] \cong  r' [\vec{b'}/x]$
	then $r\cong r'$ and $\vec{b}\cong \vec{b}'$.
\end{cor}

Given $r\lhd e$, $\vec b\lhd \bar t$ and
$e'\in \supp(\nsubst ex{\bar t})$ such that $\subst rx{\vec b}\lhd e'$,
we are about to determine the coefficient of $e'$ in $\nsubst ex{\bar t}$
by enumerating the permutations $\sigma$
such that $\subst rx{\lact\sigma{\vec b}}\lhd e'$,
\emph{i.e.} $\subst rx{\lact\sigma{\vec b}}\cong \subst rx{\vec b}$.
We thus define $\subSt{x}(r,\vec{b}) \eqdef
\{ \sigma \in \permutations{n_x(r)} \mid r[\vec{b}/x]\cong r[\lact{\sigma}{\vec{b}}/x]\} $
whenever $ \length{\vec b} = n_{x}(r)$.

\begin{lem}\label{lem:cardSvar}
	Let $ r \in\rigidExprs$ and $ \vec{b}\in D^{!}$ with $ \length{\vec b} = n_{x}(r)$.
	If $r\coh r$ then
  $ \subSt{x}(r,\vec{b}) = \restr{\G(r)}{x} \qSt(\vec{b})$.
\end{lem}
\begin{proof}
	Let $\tau\in \qSt(\vec b)$: by definition, we obtain
	$\vec\beta\in\vD(\vec b,\lact{\tau}{\vec b})$.
	If moreover $ \epsilon \in \G(r)$ then, by Lemma~\ref{lem:substPerm},
	$\subst\epsilon x{\vec \beta}\in
  \G(\subst rx{\vec b},\subst rx{\lact{\restr{\epsilon}{x}\tau}{\vec b}})$
  hence $\restr{\epsilon}{x}\tau\in \subSt{x}(r,\vec{b})$.
	It remains only to show that the function
	$(\epsilon,\tau)\in\G(r)\times \qSt(\vec b)\mapsto
  \restr{\epsilon}{x}\tau\in \subSt{x}(r,\vec{b})$
	is surjective.

	If $\sigma\in\subSt{x}(r,\vec{b})$, there exists $\phi\in
	\G(\subst rx{\vec b},\subst rx{\lact{\sigma}{\vec b}})$.
	Since $r\coh r$, we can apply Lemma~\ref{lem:substSurj} and obtain $\epsilon\in\G(r)$
	and $\vec\beta\in\vD(\vec b,\lact{\restr{\epsilon}{x}^{-1}\sigma}{\vec b})$:
	in particular, $\restr{\epsilon}{x}^{-1}\sigma\in \qSt(\vec b)$,
	and we conclude since $\sigma=\restr{\epsilon}{x}(\restr{\epsilon}{x}^{-1}\sigma)$.
\end{proof}

Our argument will moreover rely on the following construction:
if $ \length{\vec b} = n_{x}(r)$, we set
$\preqSt{x}(r,\vec b)
\eqdef\{ \epsilon \in \G(r) \mid \restr{\epsilon}{x} \in \qSt(\vec{b}) \}
=F_x^{-1}(\qSt(\vec b))$,
which is a subgroup of $\G(r)$ because
$F_x$ is a group homomorphism from $\G(r)$ to $\permutations{n_x(r)}$
by Lemma~\ref{lem:functor}.

\begin{lem}\label{lem:cardSubst}
	Let $r\in D^{(!)}$ and $\vec{b} \in D^{!}$ with $\length{\vec b}=n_x(r)$.
	If $r\coh r$ then
	$\card(\G(r[\vec{b}/x])) = \card(\preqSt{x}(r,\vec b)) \card(\vD(\vec{b}))$.
\end{lem}
\begin{proof}
	By Lemma~\ref{lem:substPerm}, if $ \epsilon \in \G(r)$
	and $\vec{\beta} \in \vD(\vec{b}, \lact{\restr{\epsilon}{x}^{-1}}{\vec{b}})$
	then
	$\epsilon [\vec{\beta} / x] \in \G(r[\vec{b} / x])$.
	If moreover $\epsilon\in\preqSt{x}(r,\vec b)$ then
	$\restr{\epsilon}{x}^{-1} \in \qSt(\vec{b})$:
	as already remarked in the proof of Lemma~\ref{lem:cardGmonomial},
	this entails that
	$\card(\vD(\vec{b}, \lact{\restr{\epsilon}{x}^{-1}}{\vec{b}}))=\card(\vD(\vec{b}))$.
	It is thus sufficient to establish that the substitution operation
	$(\epsilon, \vec{\beta}) \mapsto \epsilon[\beta / x] $ defines a bijection from
	$\sum_{\epsilon\in\preqSt{x}(r,\vec b)}\vD(\vec{b}, \lact{\restr{\epsilon}{x}^{-1}}{\vec{b}})$
	to $ \G(r[\vec{b}/ x])$.
	This fact derives immediately from Lemma~\ref{lem:substInj} (injectivity) and Lemma~\ref{lem:substSurj} (surjectivity).
\end{proof}

\begin{lem}%
	\label{lem:cardStree}
	Let $ r \in\rigidExprs$ and $ \vec{b}\in D^{!}$ with $r\coh r$ and $ \length{\vec b} = n_{x}(r)$.  Then
	\[
		\card(\subSt{x}(r,\vec b))
		= \frac{\card(\G(r))\card(\Dm(\vec b))}
		{\card(\G(\subst rx{\vec b}))}
	\quad.\]
\end{lem}
\begin{proof}
	Write $k=n_x(r)$.
  We know that $\qSt(\vec b)$ and $\restr{\G(r)}{x}$ are subgroups of $\permutations{k}$.
	Lemma~\ref{lem:cardSvar} and Fact~\ref{fact:groupProduct} entail that
	\[
		\card(\subSt{x}(r,\vec b))=
    \frac{\card(\restr{\G(r)}{x}) \card(\qSt(\vec{b}))}{\card(\restr{\G(r)}{x} \cap \qSt(\vec{b}))}\quad.
	\]
	Using Lemma~\ref{lem:cardSubst}, it will thus be sufficient to prove:
	\[
		\frac{\card(\G(r))\card(\Dm(\vec b))}
		{
			\card(\preqSt{x}(r,\vec b))
			\card(\vD(\vec{b}))
		}
		=
    \frac{\card(\restr{\G(r)}{x}) \card(\qSt(\vec{b}))}
    {\card(\restr{\G(r)}{x} \cap \qSt(\vec{b}))}
	\]
	which simplifies to
	\[
		\frac{\card(\G(r))}{\card(\preqSt{x}(r,\vec b))}
		=
    \frac{\card(\restr{\G(r)}{x})}{\card(\restr{\G(r)}{x} \cap \qSt(\vec{b}))}
	\]
	by Lemma~\ref{lem:cardGmonomial}.
	We conclude by Fact~\ref{fact:cardQuotient},
  recalling that
  $\restr{\G(r)}{x}=F_x(\G(r))$ and
  $\preqSt{x}(r,\vec b)=F_x^{-1}(\qSt(\vec b))$.
\end{proof}

\begin{lem}%
	\label{lem:coeffNsubst}
Let $e\in \resExprs$ be such that $e \coh e$ and let $\bar{t}\in \resMonomials$.
If $e'\in \supp( \partial_{x} e \cdot \bar{t})$
then $(\partial_{x} e \cdot \bar{t})_{e'} = \dfrac{m(e)m(\bar{t})}{m(e')}$.
\end{lem}
\begin{proof}
	Let $r' \lhd e'$ and $k=n_x(e)$.
	By Lemma~\ref{lem:lhdDeSubst} there exists $r \lhd e $ and
	$\vec{b} \lhd \bar{t}$ such that $r'= r[\vec{b}/x]$.
	Then, by Lemma~\ref{lem:softifySubst}, $(\partial_{x} e \cdot \bar{t})_{e'}
	= \card(\{ \sigma \in \permutations{k} \mid r[\lact{\sigma}{\vec{b}}/x]\lhd e'\})
  =\card(\subSt{x}(r,\vec b))$.
	Then we conclude by Lemmas~\ref{lem:cardStree}
	and~\ref{lem:m_is_cardG}.
\end{proof}

\section{Normalizing the Taylor expansion}%
\label{sect:commutation}

In this final section we leverage our results on the groupoid of rigid
expressions and permutation expressions in order to achieve
Steps~\ref{step:disjoint} and~\ref{step:coefNF}.
This allows us to complete the proof of commutation between Taylor expansion
and normalization.

\subsection{Normalizing resource expressions in a uniform setting}
Lemma~\ref{lem:coeffNsubst} is almost sufficient to obtain Step~\ref{step:coefNF},
as it fixes the coefficients in a hereditary head reduction step from a uniform expression:
\begin{lem}%
	\label{lem:coeffL}
Let $e\in \resExprs$ with $e \coh e$. If $e'\in \supp(L(e))$ then $(L(e))_{e'}= \dfrac{m(e)}{m(e')}$.
\end{lem}

\begin{proof}
	By induction on the structure of $e$ applying
	Lemma~\ref{lem:coeffNsubst}
	in the redex case: observe indeed that if
	$e=\lambda\vec x. \rappl{\lambda y. s}{\bar t\,{\bar u_1\cdots\bar u_k}}$
	then $e'=\lambda\vec x. \rappl{v}{\bar u_1\cdots\bar u_k}$
	with $v\in \supp(\nsubst sy{\bar t})$,
	and then $(L(e))_{e'}=(\nsubst sy{\bar t})_v=\frac{m(s)m(\bar t)}{m(v)}$
	and we conclude since $\frac{m(e)}{m(e')}=\frac{m(s)m(\bar t)}{m(v)}$.
	All the other cases follow directly from the induction hypothesis by multilinearity.
\end{proof}

To iterate Lemma~\ref{lem:coeffL} along the reduction sequence to the normal form,
we first need to show that uniformity is preserved by $L$.
As before, we prefer to focus on the rigid setting first, and
we will only consider the hereditary head reduction defined as follows:\footnote{
  Note that the reduction from $\rappl{\lambda x.a}\vec b$ to $\subst ax{\vec b}$
  is not well behaved in general: its contextual extension is not even confluent,
  because it forces the order in which variable occurrences are substituted.
  Consider for instance the term $(\lambda x.\rappl{\lambda y.\rappl y{(x)}}(x))(z_1,z_2)$
  which has two distinct normal forms: $\rappl{z_1}{z_2}$ and $\rappl{z_2}{z_1}$.
  This rigid calculus is thus not very interesting \emph{per se},
  and we only consider it as a tool to analyze the dynamics of the resource calculus.
}
\begin{gather*}
	\begin{aligned}
		L(\linj a)
		&\eqdef \linj{L(a)}
		&
		L(\rinj a)
		&\eqdef \rinj{L(a)}
		\\
		L( \lambda\vec x.\lambda y. (\linj a))
		&\eqdef\lambda\vec x.(\linj{\lambda y.a})
		&
		L( \lambda\vec x.\lambda y. (\rinj a))
		&\eqdef\lambda\vec x.(\rinj{\lambda y.a})
	\end{aligned}
	\\
	\begin{aligned}
		L(\lambda\vec x.\rappl{\rappl{\linj {a}}{\vec b}}{\vec c_1\cdots\vec c_k})
		&\eqdef\lambda\vec x.\rappl{\linj {\rappl{a}{\vec b}}}{\vec c_1\cdots\vec c_k}
		\\
		L(\lambda\vec x.\rappl{\rappl{\rinj {a}}{\vec b}}{\vec c_1\cdots\vec c_k})
		&\eqdef\lambda\vec x.\rappl{\rinj {\rappl{a}{\vec b}}}{\vec c_1\cdots\vec c_k}
		\\
		L(\lambda\vec x.\rappl{y}{\vec{a}_{1}\cdots \vec{a}_{k}})
		&\eqdef \lambda\vec x.\rappl{y}{L(\vec{a}_{1})\cdots L(\vec{a}_{k})}
		\\
		L((a_{1}, \dots , a_{k}))
		&\eqdef (L(a_{1}), \dots ,L(a_{k})  )
		\\
		L(\lambda\vec x. \rappl{\lambda y. a}{\vec b\,{\vec c_1\cdots\vec c_k}})
		&\eqdef\lambda\vec x.\rappl{\subst a y{\vec b}}{{\vec c_1\cdots\vec c_k}}
	\end{aligned}
\end{gather*}
extended to partial rigid expressions by setting $L(0)\eqdef0$.
By an analogue of Lemma~\ref{lem:SN}, for any $r\in\rigidExprs$, there exists $k\in\N$ such that
$L^k(r)$ is normal, and then we write $\NF r=L^k(r)$.
Moreover, $r$ is in normal form iff $L(r)=r$.

\begin{lem}\label{lem:redSupp}
	If $e \in \resExprs$ then:
	\begin{enumerate}
		\item $\supp(L(e))=\{\softify{L (r)}\mid r\lhd e\text{ and }L(r)\not=0\}$;
		\item $\supp(\NF e)=\{\softify{\NF r}\mid r\lhd e\text{ and }\NF r\not=0\}$.
	\end{enumerate}
\end{lem}
\begin{proof}
	We first prove that $r'\lhd e'\in \supp(L(e))$
	iff there exists $r\lhd e$ with $r' = L(r)$,
	which gives the first result:
	this is done by a straightforward induction on the structure of $e$,
	using Corollary~\ref{cor:lhdSubst} for the $\beta$-redex case.

	Now fix $k\in\N$ such that $\NF e=L^k(e)$:
	by iterating the previous result, we obtain
	$r'\lhd e'\in  \supp(\NF e)$
	iff there exists $r\lhd e$ with $r' = L^k(r)$.
	Then we conclude, observing that if $r'\lhd e'$,
	then $r'$ is in normal form iff $e'$ is.
\end{proof}

\begin{lem}%
	\label{lem:cohSubst}
		If $r\coh r'$ and $\vec b\coh \vec b'$
		with $n_x(r)=\length{\vec b}$ and $n_x(r')=\length{\vec b'}$
		then $\subst rx{\vec b}\coh \subst {r'}x{\vec b'}$.
\end{lem}
\begin{proof}
	By a straightforward induction on the derivation of $r\coh r'$.
\end{proof}

\begin{lem}%
	\label{lem:cohRedRigid}
	For all $r,r'\in\rigidExprs$ such that $r\coh r'$:
	\begin{enumerate}
		\item if $L(r)\not=0$ and $L(r')\not=0$ then $L(r)\coh L(r')$;
		\item if $\NF r\not=0$ and $\NF{r'}\not=0$ then $\NF r\coh \NF{r'}$.
	\end{enumerate}
\end{lem}
\begin{proof}
	The first item is easily established by induction on $r$,
	using Lemma~\ref{lem:cohSubst} in the case of a $\beta$-redex.
	Having fixed $k$ such that both $\NF r=L^k(r)$ and $\NF{r'}=L^k(r')$,
	the second item follows by iterating the first one.
\end{proof}

We have thus established that $L$ preserves coherence of rigid expressions.
It follows that $L$ preserves cliques of resource expressions:

\begin{lem}%
	\label{lem:cohRed}
	If $E\subseteq\resExprs$ is a clique, then
  both $L(E)$ and $\NF E$ are cliques.
\end{lem}
\begin{proof}
	As a direct consequence of Lemmas~\ref{lem:redSupp} and~\ref{lem:cohRedRigid},
  we obtain that:
  if $e\coh e'$ then,
  for all $e_0\in \supp(L(e))$ and  $e'_0\in \supp(L(e'))$
  (resp. $e_0\in \supp(\NF e)$ and  $e'_0\in \supp(\NF{e'})$),
  we have $e_0\coh e'_0$.
  The result follows straightforwardly.
\end{proof}

Step~\ref{step:disjoint} amounts to the fact that distinct coherent expressions
have disjoint normal forms.
In other words, if the normal forms of two coherent expressions intersect on a
common element, then they must coincide.
This result will follow from the following rigid version,
which states that coherent rigid expressions with isomorphic normal forms are isomorphic:
\begin{lem}%
	\label{lem:suppNFRigid}
	For all $r,r'\in\rigidExprs$ such that $r\coh r'$:
	\begin{enumerate}
		\item if $L(r)\cong L(r')$ then $r\cong r'$;
		\item if $\NF r\cong \NF{r'}$ then $r\cong r'$.
	\end{enumerate}
\end{lem}
\begin{proof}
	Observe that $\cong$ is defined on rigid expressions only
	so that if, e.g., $L(r)\cong L(r')$
	then in particular $L(r)\not =0\not=L(r')$.
	The first item is established by induction on $r$,
	using Corollary~\ref{cor:cohAntired} in the case of a $\beta$-redex.
	Having fixed $k$ such that both $\NF r=L^k(r)$ and $\NF{r'}=L^k(r')$,
	the second item follows by iterating the first one,
	thanks to Lemma~\ref{lem:cohRedRigid}.
\end{proof}

Note that the converse does not hold, even in the uniform case:
two uniform, isomorphic and coherent rigid expressions may yield
normal forms that are not isomorphic.
\begin{exa}
Consider
$a=\rappl{\lambda x. \rappl{x} {(x)} }{(y \oplus \bullet, \bullet \oplus z)}$ and
$a'=\rappl{\lambda x. \rappl{x} {(x)}} (\bullet \oplus z, y \oplus \bullet)$.
We have $a\coh a'$ and $a\cong a'$ but
$\NF a=L(a)= \rappl{y \oplus \bullet}{(\bullet \oplus z)}$
and
$\NF{a'}=L(a')= \rappl{\bullet \oplus z}{(y \oplus \bullet)}$,
hence $\NF a\not\cong\NF{a'}$.
\end{exa}

\begin{thm}[Step~\ref{step:disjoint}]\label{thm:disjoint}
  Let $e, e' \in\resExprs$ be such that $e\coh e'$.
  If $\supp(\NF{e})\cap \supp(\NF{e'})\not=\emptyset$ then $e=e'$.
\end{thm}
\begin{proof}
	Let $e_0\in \supp(\NF{e})\cap \supp(\NF{e'})$.
	By Lemma~\ref{lem:redSupp},
	there are $r\lhd e$ and $r'\lhd e'$
  such that $e_0=\softify{\NF{r}}=\softify{\NF{r'}}$.
	Since $e\coh e'$, we have $r\coh r'$ and,
	since ${\NF{r}}\cong{\NF{r'}}$, we obtain $r\cong r'$ by Lemma~\ref{lem:suppNFRigid},
	hence $e=e'$.
\end{proof}

By Lemma~\ref{lem:cohRed}, $L$ preserves coherence;
and thanks to Theorem~\ref{thm:disjoint} we can iterate
Lemma~\ref{lem:coeffL} to obtain:
\begin{thm}[Step~\ref{step:coefNF}]%
	\label{thm:coeffNF}
	Let $e\in \resExprs$ with $e\coh e$ and let $e'\in \supp(NF(e))$.
	Then \[(\NF{e})_{e'} = \dfrac{m(e)}{m(e')}\quad .\]
\end{thm}
\begin{proof}
  Fix $n$ such that $\NF{e}=L^n(e)$:
  since $\supp(\NF{e})=L^n(\{e\})$,
  there exists a sequence $e_0,\dotsc,e_n$
	such that $e_0=e$, $e_n=e'$
  and $e_{i}\in \supp(L(e_{i-1}))$
	for $1\le i\le n$.
  We prove by induction on $n$ that,
  given such a sequence, we have
  $\NF{e_0}_{e_n}={m(e_{0})}/{m(e_n)}$.

  If $n=0$ the result is trivial.
  Otherwise, Lemma~\ref{lem:coeffL} gives
  $L(e_0)_{e_{1}}={m(e_{0})}/{m(e_{1})}$.
  Moreover, Lemma~\ref{lem:cohRed} ensures that $\supp(L(e_0))$ is a clique,
  and in particular $e_1\coh e_1$ and the induction hypothesis entails
  $\NF{e_1}_{e_{n}}={m(e_{1})}/{m(e_{n})}$.
  Finally, since $e_n\in \supp(\NF{e_1})$, Theorem~\ref{thm:disjoint}
  entails $\NF{e'}_{e_n}=0$ for each $e'\in \supp(L(e_0))\setminus\{e_1\}$.
  We obtain
  \[
    \NF{e_0}_{e_n}=\NF{L(e_0)}_{e_n}=L(e_0)_{e_1}\NF{e_1}_{e_n}
    =\dfrac{m(e_0)}{m(e_1)}\dfrac{m(e_1)}{m(e_n)}
    =\dfrac{m(e_0)}{m(e_n)}
    \quad.
    \qedhere
  \]
\end{proof}

\subsection{Commutation}

By assembling all our previous results, we obtain the desired commutation
theorem:

\begin{thm}
Let $M\in \ndTerms$. Then $ \tayexp{\BT{M}} = \NF{\tayexp{M}}$.
\end{thm}
\begin{proof}

 By Theorem~\ref{thm:m}  \[\tayexp{M} = \sum\limits_{s\in T(M)} \frac{1}{m(s)} s\]
 and by Theorem~\ref{thm:clique} and Theorem~\ref{thm:disjoint} we are allowed to form
 \[
	 \NF{\tayexp{M}}
	 \quad
	 = \sum\limits_{s\in\taysup M} \frac{1}{m(s)} \NF s
	 \quad
	 = \sum\limits_{s\in\taysup M} \sum\limits_{u\in \supp(\NF s)} \frac{\NF s_u}{m(s)}  u
 \]
 the inner sums having pairwise disjoint supports.
 Then, if $u\in \supp(\NF{\tayexp M})$,
 there is a unique $s\in\taysup M$ such that $u\in \supp(\NF s)$
 and we obtain
 $\NF{\tayexp{M}}_u = \frac{\NF s_u}{m(s)} = \frac{1}{m(u)}$
 by Theorem~\ref{thm:coeffNF}.
 We conclude since $\supp(\NF{\tayexp M})=\taysup{\BT M}$  by Theorem~\ref{thm:taysup}.
\end{proof}

\section*{Acknowledgements}
This work owes much to the friendly and stimulating environment provided by the
International Research Network on Linear Logic\footnote{\url{http://www.linear-logic.org/}}
between the French CNRS and the Italian INDAM\@.

\bibliographystyle{alphaurl}
\bibliography{group}
\end{document}